\documentclass[11pt,a4paper]{article}

\usepackage{fullpage}
\usepackage{graphicx,amsmath,amssymb,amsthm}

\usepackage{todonotes}

\newtheorem{theorem}{Theorem}[section]

\newtheorem{lemma}[theorem]{Lemma}
\newtheorem{corollary}[theorem]{Corollary}

\newtheorem{definition}[theorem]{Definition}
\newtheorem{observation}[theorem]{Observation}

\newtheorem{calculation}{Calculation}

\DeclareMathOperator{\supp}{Supp}
\DeclareMathOperator{\dist}{dist}

\newcommand{\ignore}[1]{}

\title{Trading query complexity for sample-based testing and multi-testing scalability}
\author{Eldar Fischer\thanks{Faculty of Computer Science, Israel Institute of Technology (Technion), Haifa, Israel. \mbox{eldar@cs.technion.ac.il}} \and Oded Lachish\thanks{Birkbeck, University of London, London, UK. \mbox{oded@dcs.bbk.ac.uk}}
\and Yadu Vasudev\thanks{Faculty of Computer Science, Israel Institute of Technology, Haifa, Israel. \mbox{yaduvasudev@gmail.com}}}
\begin{document}
\maketitle
\begin{abstract}
We show here that every non-adaptive property testing algorithm making a constant number of queries, over a fixed alphabet, can be converted to a sample-based (as per [Goldreich and Ron, 2015]) testing algorithm whose average number of queries is a fixed, smaller than $1$, power of $n$. Since the query distribution of the sample-based algorithm is not dependent at all on the property, or the original algorithm, this has many implications in scenarios where there are many properties that need to be tested for concurrently, such as testing (relatively large) unions of properties, or converting a Merlin-Arthur Proximity proof (as per [Gur and Rothblum, 2013]) to a proper testing algorithm.

The proof method involves preparing the original testing algorithm for a combinatorial analysis, which in turn involves a new result about the existence of combinatorial structures (essentially generalized sunflowers) that allow the sample-based tester to replace the original constant query complexity tester.
\end{abstract}

\section{Introduction}\label{sec:intro}
A {\em test} for a property $L\subseteq \Xi^n$ (where $\Xi$ is a fixed alphabet), with proximity parameter $\epsilon$, is an algorithm that queries an input $w\in\Xi^n$ in a limited number of places, and distinguishes with high probability between the case that $w\in L$ and the case that no $w'\in L$ is $\epsilon$-close to $w$ in the normalized Hamming distance. A {\em non-adaptive} test is a test that decides its queries in advance of receiving the corresponding input values, which basically means that its queries are governed by a single distribution $\mu$ over the power set of $[n]$.

Given a family of properties $\mathcal F \subseteq \{L:L\subseteq \Xi^n\}$, we say that there is a {\em canonical} testing scheme for $\mathcal F$ if there are non-adaptive tests (with the same parameter $\epsilon$) for all members $L\in\mathcal F$, which additionally all share the same query probability distribution $\mu$.

This concept has been defined and used before. The most well-known example is that of \cite{ThreeTheorems}, where the family of all properties of dense graphs (as per the model defined in \cite{GGR}) with $n$ vertices that are testable (non-adaptively or not) with up to $q$ queries, is shown to have a canonical testing scheme, where the common query distribution consists of uniformly picking a set of $q$ vertices and querying all $\binom{q}{2}$ vertex pairs.

Note that being in the dense graph model in essence restricts the admissible properties. Under this model, an input $w\in\{0,1\}^{\binom{n}{2}}$ is interpreted as the adjacency matrix of a graph with $n$ vertices, and a property $L$ is admissible if it is invariant under all input transformations corresponding to re-labeling the graph vertices (i.e., the transformations corresponding to graph isomorphisms).

There are other examples. For example, given a finite field $F$, for properties of functions over a linear space over $F$ that are known to be invariant under linear transformations, a canonical testing scheme would consist of querying the function over an entire small dimensional subspace picked uniformly at random \cite{LinearInv}.

A natural question is what would be a candidate for an ``ultimate'' canonical scheme, where there are no structural impositions on the property at all. One would expect here a query distribution that is completely symmetric with respect to any permutation of the index set $[n]$. Indeed, such a scheme is defined as {\em sample-based} testing in \cite{SampleBased}. The sampled-based distribution $\mu_p$ corresponds to choosing every index $i\in [n]$ to be queried independently with probability $p$. Usually $p$ will be $n^{-\alpha}$ for some $1>\alpha>0$. It is a folly to expect a sample-based testing scheme with significantly fewer queries, even for properties with a constant bound on the number of queries for a test, as evidenced already in \cite[Proposition 6.9]{GGR}.

In \cite{SampleBased} a connection between {\em proximity-oblivious} testers (POT) as defined in \cite{ProximityOblivious} and the sample-based querying scheme was suggested. Proximity oblivious testers are non-adaptive testing algorithms whose querying distribution is the same for any proximity parameter $\epsilon$, where instead the distinguishing probability between inputs in $L$ and inputs $\epsilon$-far from $L$ changes with $\epsilon$. The work \cite{SampleBased} showed that for such testers that additionally have the property that all indexes get queried with about the same probability (but not necessarily in an independent manner), there exists a conversion to sample based testers with $p=O(n^{-1/q})$, where the coefficient depends on the distinguishing probability, and the parameter measuring the above-mentioned ``probability sameness'' of the original test.

In \cite{fischer2014partial} it is shown that all $1$-sided proximity-oblivious testers over the alphabet $\{0,1\}$ are convertible to the canonical sample-based scheme, where $p=n^{-\alpha}$ with $\alpha$ depending (somewhat badly) on $q$, $\epsilon$ and the distinguishing probability $\delta$. In \cite{SampleBased} there is an example of a testable property that has no sublinear query complexity sample-based test at all, but it works only over an alphabet whose size is exponential in $n$, and so does not contradict the result of \cite{fischer2014partial}.

Here we take the investigation much further, and prove the following.

\begin{theorem}[informal statement of our main result]\label{thm:informal}
Every property of words in $\Xi^n$ that has a non-adaptive $\epsilon$-test with $q$ queries and detection probability $\delta$ (either $1$-sided or $2$-sided) admits a test using the sample-based canonical querying scheme, where the distribution $\mu_p$ has $p=O(n^{-\alpha})$, with $\alpha$ depending on $q$, $\delta$, and for $2$-sided testing also on $|\Xi|$ and $\epsilon$, and the hidden coefficient depending on $q$, $\delta$, $\epsilon$ and $|\Xi|$.
\end{theorem}

We prove this separately for $1$-sided tests and $2$-sided tests. For $2$-sided tests we go further and prove the result for {\em partial tests}, that are only guaranteed to accept inputs in some sub-property $L'$ with high probability, a generalization whose relevance is explained below.

For both $1$-sided testing and $2$-sided (possibly partial) testing we obtain a very improved bound on $\alpha$ as compared to the $1$-sided testing result of \cite{fischer2014partial}. Additionally, the dependency of the coefficient on $|\Xi|$ is logarithmic, while for the $2$-sided test the additional dependency of $\alpha$ on $|\Xi|$ is of type $\log\log\log(|\Xi|)$. This shows that the exponential size of the alphabet in the counter-example in \cite{SampleBased} is essential.

We believe that the ``correct'' $\alpha$ should be just $-1/q$, at least for converting $1$-sided tests, but cannot prove it yet.

\subsection{Implications for multitests}

There are several motivations for finding canonical testing schemes. One of them is for proving lower bounds, which may be easier when the querying distribution is ``simple'' and known. Here we would like to highlight another one, which also played an implicit role in the original motivation of \cite{fischer2014partial}.

Given a sequence of properties $L_1,\ldots,L_r$, a {\em multitest} for them is an algorithm that makes queries to a word $w\in\Xi^n$, and provides a sequence of answers. With probability at least $1-\delta$, the answers should be correct for {\em all} the properties, that is, for every $k$ such that $w\in L_k$ the corresponding answer should be ``yes'', and for every $l$ such that $w$ is $\epsilon$-far from $L_k$ the corresponding answer should be ``no''.

If we know nothing else about the properties apart from that they are all testable using $q$ queries each, then the scalability of a test to a multitest would be quasilinear: We first take a test for every $L_i$, and amplify its success probability to $1-\delta/r$ (which multiplies the number of queries by $O(\log r)$). Then we just run these $r$ tests one after the other, and use the union bound for the total success probability, all in all using $O(q\cdot r\log r)$ queries. This is not always good enough, as in some applications $r$ can depend on $n$, and may even be greater than $n$ (say through a polynomial dependency).

However, the situation changes dramatically if we know all properties to share a canonical testing scheme with $q'$ queries (where $q'$ could depend on $n$). In this case, we can re-use the same queries for all $r$ (amplified) tests, and the union bound will still work. This brings us to using only $O(q'\cdot \log r)$ queries in all. This scalability can have many implications.

In \cite{fischer2014partial}, multitests are implicitly used for testing unions of properties. This in turn allows to convert in certain cases tests requiring proofs as per the $\mathcal{MAP}$ scenario (defined in \cite{MAP} and also developed in \cite{fischer2014partial}) to tests that still have a sublinear query complexity but do not require such proofs. In this setting we deploy the generalization of our result to partial testing, as a $\mathcal{MAP}$ scenario converts to a union of partial testing problems.

Another scenario aided by a multitest is if one wants to store the results for $w$ belonging (approximately) to a rather large set of possible properties. If the properties share a canonical testing scheme, and the corresponding property tests also admit a not too large computation time overhead, then it may be worthwhile to store instead the common set of queries performed by the multitest, because this query set increases rather slowly with $r$.

Finally, a canonical testing scheme also allows for some measure of privacy: Suppose that one wants to test a particular property of $w\in\Xi^n$, but wants to hide from the ``input holder'' the identity of the particular property to be tested. By using the canonical scheme, no one but the party performing the test can discern which of the properties having the canonical scheme is being tested for.

\subsection{Methods used}

The crucial analysis used for converting a test with $q$ queries to a sample-based test is of a combinatorial nature. We take the support of the query distribution of a non-adaptive test, and analyze it as a family of query sets, essentially a $q$-uniform hypergraph whose vertex set is the domain of possible queries.

For $1$-sided testing algorithms, since they reduce to checking whether the set of queries is a witness refuting the possibility of the input belonging to the property, the support of the distribution provides most of the information we need. We can assume (through a simple processing of the original test) that the number of possible query sets is linear in the domain size $n$. Finding large ``matchings'' (families of disjoint sets) of refuting witnesses would be ideal for sample based testing, but since we make no assumptions on our family of sets outside its size, these do not always exist.

The next option that could be explored is finding large sunflowers as defined in \cite{sunflower}. This is the approach taken by \cite{fischer2014partial}, and it can be generalized to the setting here. However, the obtained $\alpha$ value for the $n^{-\alpha}$ sampling test would depend very badly on the other parameters, because sunflowers require processing in several stages.

Here we present a generalization of sunflowers, which we call {\em pompoms}. The main difference is that the ``core'' common to the participating sets is not their intersection as in sunflowers, and in fact could be much larger -- the only requirement is that the participating sets are disjoint outside the core. The support of the query distribution is shown to admit pompoms larger than the sunflowers it would admit, and moreover ones that all share the same core, so we can rule out the possible inputs in just one processing round.

For $2$-sided testing we use pompoms to explicitly estimate what some portions of the original test would say for a particular input. Because of the estimation requirements, we need an additional requirement of the $2$-sided test -- it has to be {\em combinatorial}, in the sense that its query distribution is uniform over the family of possible query sets. Much work is needed to fully convert a general $2$-sided test to a combinatorial one that can be analyzed as a hypergraph, and this introduces some extra dependency on the alphabet size. To aid with the analysis of the $2$-sided tests, a formalism of {\em probabilistic formulas} is introduced. The combinatorialization shown here has potential for future uses, as it generalizes to promise problems apart from property testing -- it is enough to have at least one ``yes'' instance and one ``robust no'' instance.

The pompom families used both in the $1$-sided and the $2$-sided sample-based testing result both stem from a general structural result about the query distributions, that finds a {\em constellation} in the underlying support. This in turn readily allows for the extraction of pompoms.

\section{Preliminaries}\label{sec:prelim}

\subsection{Large deviation bounds}

The following is useful for the analysis of sample based testing.
\begin{lemma}[multiplicative Chernoff bounds]\label{lem:chernoff}
Let $X_1,\dots,X_m$ be i.i.d $0$-$1$ random variables such that $\Pr[X_i=1]=p$. Let $X=\sum_{i=1}^m X_i$. For any $\gamma\in(0,1]$,
$$\Pr\left[ X>(1+\gamma) pm \right] < \exp\left( -\gamma^2 pm/3 \right)$$
$$\Pr\left[ X<(1-\gamma) pm \right] < \exp\left( -\gamma^2 pm/2 \right)$$
\end{lemma}

\begin{lemma}[Hoeffding bounds,\cite{Hoeffding63}]\label{lem:hoeff-bound}
Let $Y_1,\cdots,Y_m$ be independent random variables such that $0\leq Y_i\leq 1$, for $i=1,\cdots,m$ and let $\eta=\mathrm E\left[\sum_{i=1}^m Y_i/m\right]$. Then,
$$\Pr\left[ \left|\frac{\sum_{i=1}^m Y_i}{m}-\eta\right| \geq t \right]
\leq 2\exp(-2mt^2)$$
\end{lemma}

\begin{lemma}[without replacement,\cite{Hoeffding63}]\label{lem:replacement}
Let $X_1,\ldots,X_m$ be random variables picked uniformly without repetition from the sequence $\mathcal{C} = (\gamma_1,\ldots,\gamma_m)$ where $0\leq\gamma_i\leq 1$ (this means that $i_1,\ldots,i_m$ are picked uniformly without repetition from $[m]$, and then every $X_j$ is set to $\gamma_{i_j}$; it may be that some $\gamma_i$ is equal to another). Let $Y_1,\cdots,Y_m$ be independent random variables picked with repetition from $\mathcal{C}$ (i.e. every $k_j$ is uniformly and independently chosen from $[m]$ and then $Y_j$ is set to $\gamma_{k_j}$). Then the conclusion of Lemma \ref{lem:hoeff-bound} for $Y_1,\ldots,Y_m$ holds also for $X_1,\ldots,X_m$, that is,
$$\Pr\left[ \left|\frac{\sum_{i=1}^m X_i}{m}-\eta\right| \geq t \right]
\leq 2\exp(-2mt^2)$$
where $\eta=\frac1m\sum_{i=1}^m\gamma_i$.
\end{lemma}

\begin{lemma}[large deviation bound]\label{lem:dev}
Denote by $\mu_p$ the distribution over subsets of $[m]$, where every $i\in [m]$ is picked into the subset with probability exactly $p$, independently from all other $j\neq i$.
Suppose that $\gamma_1,\ldots,\gamma_m$ are all values in $[0,1]$, and let $U\subseteq [m]$ be chosen according to $\mu_p$, where $p\geq 10c/\eta^2 m$ and $c>1$. Then, with probability at least $1-e^{-c}$, the value $(\sum_{i\in U}\gamma_i)/|U|$ (where we arbitrarily set it to $\frac12$ if $U=\emptyset$) is in the range $(\sum_{i=1}^m\gamma_i)/m\pm\eta$.
\end{lemma}

\begin{proof}
If $U$ is picked according to $\mu_p$, then $\mathrm E[|U|]=pm$. By the multiplicative Chernoff bound in Lemma \ref{lem:chernoff} we have the following bound on the probability of the size of $U$ being small:
$$\Pr\left[|U|\leq\frac{pm}{2}\right]\leq \exp(-pm/8)\le e^{-c}/2.$$

We continue our analysis conditioned on the event that the size of $U$ is at least $pm/2$. For every $k\geq pm/2$, let us analyze separately the deviation of the value $\left( \sum_{i\in U}\gamma_i \right)/|U|$ conditioned on $|U|=k$. Lemma \ref{lem:replacement} holds for this case, stating that the probability of $\left|\frac1k\sum_{i=1}^k X_i-\frac1m\sum_{i=1}^m \gamma_i\right|$ being greater than $\eta$ is bounded by $e^{-c}/2$. Hence, for $U$ picked according to $\mu_p$, the probability of $\sum_{i\in U}\gamma_i/|U|$ being outside the range $\sum_{i\in m}\gamma_i/m\pm \eta$ is at most $e^{-c}$, by the union bound on the event of $|U|<pm/2$, and the event of $\left|\frac1k\sum_{i=1}^k X_i-\frac1m\sum_{i=1}^m \gamma_i\right|>\eta$ while $|U|\geq pm/2$.
\end{proof}

\subsection{Words and distributions}

\paragraph{Notation for words}
Let $\Xi$ be an alphabet, and let $w \in \Xi^n$, $i \in [n]$ and $Q\subseteq [n]$. We use $w_i$ to denote the $i$'th letter of $w$ and $w_Q$ to denote the word $v \in \Xi^{|Q|}$ such that, for every $j\in [|Q|]$, $v_j=w_{Q(j)}$, where $Q(j)$ is the $j$'th smallest member of $Q$.
Let $C\subseteq [n]$ and $\sigma$ be a word in $\Xi^{|C|}$.
We denote by $w_{\sigma,C}$ the word that we get by taking $w$ and replacing its sub-word $w_{C}$ with $\sigma$.

\begin{definition}[word distances]
Two words $w,v\in\Xi^n$ are said to be {\em $\epsilon$-far} if there is no $A$ of size at most $\epsilon n$ for which $w_{[n]\setminus A}=v_{[n]\setminus A}$ (in other words, we use the normalized Hamming distance). Otherwise these words are said to be {\em $\epsilon$-close}. Given a property $L\subseteq \Xi^n$, a word $w$ is said to be {\em $\epsilon$-close} to $L$ if there exists an $\epsilon$-close word $v$ which is in $L$, and otherwise $w$ is said to be {\em $\epsilon$-far} from $L$.
\end{definition}

\paragraph{Notation for distributions}
We deal with distributions $\mu$ over subsets of $[n]$. For $A\subseteq [n]$ we denote by $\mu(A)$ the probability of $A$ being drawn by $\mu$. For a non-empty event, that is a family of sets $\emptyset\neq\mathcal A\subseteq 2^{[n]}$, we abuse notation somewhat and denote $\mu(\mathcal A)=\sum_{A\in\mathcal A}\mu(A)$. We denote by $\supp(\mu)$ the family of positive probability outcomes $\{A\subseteq [n]:\mu(A)>0\}$, and for two distributions $\mu$ and $\mu'$ denote by $\dist(\mu,\mu')$ the variation distance $\frac12\sum_{A\subseteq [n]}|\mu(A)-\mu'(A)|=\max_{\mathcal A\subseteq 2^{[n]}}|\mu(\mathcal A)-\mu'(\mathcal A)|$.

\subsection{Property Testing}

We start this subsection by defining tests. We define partial tests (of which tests are a special case), because we would like our main result to also have applications in the realm of $\mathcal{MAP}$s as defined in \cite{MAP}.

\begin{definition}[$(\epsilon,\delta,q)$-test]
Given two properties $L'\subseteq L\subseteq \Xi^n$,
a partial {\em $(\epsilon,\delta,q)$-test} for $(L',L)$ is a
  randomized algorithm $\mathcal{A}$ that, given 
  query access to the input $w$, uses $q$ queries and satisfies the following:
  \begin{enumerate}
    \item If $w\in L'$, 
      then $\Pr\left[ \mathcal{A}(w)=1 \right] \ge 1-\delta$.
    \item If $w$ is $\epsilon$-far from $L$, 
      then $\Pr\left[ \mathcal{A}(w)=0 \right] \ge 1-\delta$.
  \end{enumerate}
The test is {\em $1$-sided} if, when $w\in L'$, the output is always $1$, and otherwise it is {\em $2$-sided}. If the choice of every query is independent of the answers to the previous queries, then the test is {\em non-adaptive}, and otherwise it is {\em adaptive}.

In the case where $L'=L$, we call it a {\em $(\epsilon,\delta,q)$-test} for $L$.
\end{definition}

We remark that, in the case of a non-adaptive test, we may assume that the set of queries is selected before any query is made.
So, a non-adaptive test can be viewed as consisting of three steps: 
(i) a set of queries $Q$ is randomly selected according to a distribution over $2^{[n]}$; 
(ii) 
the sub-word $w_Q$ is queried;
(iii) the output is computed according to $(Q,w_Q)$.

We note that a $1$-sided test can reject only if $(Q,w_Q)$ constitutes a proof that $w$ is not in the property.
This occurs if and only if $Q$ is a {\em witness against $w$} as defined next.

\begin{definition}[witness against a word]
A set $Q\subseteq [n]$ is a {\em witness} against a word $w\in \Xi^n$ (with regards to a property $L$), if every $u\in \Xi^n$ such that $u_Q = w_Q$ is not in $L$.
\end{definition}

Without loss of generality, we assume that a test always rejects when it encounters a witness against the word. In the case of a $1$-sided tester it is actually the case that the test rejects only if it encountered a witness. We next formally define the concept of the distribution of a non-adaptive test.

\begin{definition}[distribution of a non-adaptive $(\epsilon,\delta)$-test]
The distribution of a non-adaptive $(\epsilon,\delta)$-test $\mathcal{A}$, denoted by $\mu_{\mathcal{A}}$, is a distribution over $2^{[n]}$, such that, for every $Q\subseteq [n]$, the value of $\mu_{\mathcal{A}}(Q)$ is the probability $\mathcal{A}$ will select $Q$ to be its set of queries.
We omit the subscript when it is clear from context.
\end{definition}

Our conversion results rely on the combinatorial aspects of  distributions of tests. In fact, for non-adaptive $1$-sided tests without loss of generality this distribution is the sole defining object, because the test can be assumed to reject if and only if its query set produced a witness against the input word.
In particular, we show a reduction to the case where the cardinality of the support of the distribution has a bound linear in $n$.
We use the following definition to capture this case and afterwards we give the reduction.

\begin{definition}[non-adaptive $(\epsilon,\delta,q,k)$-test]
A non-adaptive $(\epsilon,\delta,q)$-test is a {$(\epsilon,\delta,q,k)$-test}, if $|\supp(\mu)|\leq k$.
\end{definition}

We observe that the support of the distribution of an $(\epsilon,\delta,q)$-test contains only sets of cardinality $q$.
We use the term {\em $(\epsilon,\delta)$-test} (omitting $q$) when we do not make any assumption on the cardinality of the sets in the distribution. The following lemma transforms a $1$-sided test to one with parameters more suitable for analysis and conversion to sample-based testing.

\begin{lemma}\label{lem:effective1sided}
A non-adaptive $1$-sided $(\epsilon/2,\delta,q)$-test can be converted to a non-adaptive $1$-sided $(\epsilon/2,1/2(q'+1),q',4(q'+1)^2\log(|\Xi|)n)$-test where $q'=O(q\log(q)/(1-\delta))$.
\end{lemma}

\begin{proof}
First, by traditional amplification, repeating the original test $10\log(q)/(1-\delta)$ times and rejecting if any run had rejected, we convert it to an $(\epsilon/2,1/1000q'',q'')$-test where $q''=O(q\log(q)/(1-\delta))$. Then we consider the outcome of running the test $10\log(|\Xi|)q''n$ times independently. By Lemma \ref{lem:chernoff}, for any fixed $\epsilon/2$-far input $w\in\Xi^n$, the probability that it is accepted by more than a $1/10q''$ fraction of the runs is bounded by $e^{-99^2\cdot \log(|\Xi|)n/300}<\frac12|\Xi|^{-n}$. This means that with probability at least $\frac12$, such a sequence of runs will satisfy the above for all $\epsilon/2$-far inputs at once. We fix such a sequence of runs, and make it the new test. That is, the new $\mu'$ consists of selecting one of the fixed runs uniformly at random, and using its query set. This brings us to an $(\epsilon/2,1/10q'',q'',10\log(|\Xi|)q''n)$-test. We artificially increase the number of queries to $q'=3q''$ to obtain our required test.
\end{proof}

The following two lemmas are essential for the analysis of $1$-sided test.
When reading them, one should have in mind that, when they are applied, the $\delta$ parameter in their statement is small since the tests analyzed are those implied by Lemma~\ref{lem:effective1sided}.

\begin{lemma}\label{lem:SmallSupport}
Let $\mathcal{J}\subseteq \supp(\mu)$, where $\mu$ is the distribution of a $1$-sided $(\epsilon/2,\delta)$-test for a non-empty property $L$ for which $\epsilon$-far words exist.  If $|\bigcup_{Q\in \mathcal{J}}Q| < \epsilon n/2$, then $\mu(\mathcal{J}) < \delta$.
\end{lemma}

\begin{proof}
Let $T = \bigcup_{Q\in \mathcal{J}}Q$, $u\in\Xi^n$ be a word in $L$, $w\in \Xi^n$ be $\epsilon$-far from $L$, and $v\in \Xi^n$ be such that $v_T = u_T$ and $v_{[n]\setminus T} = w_{[n]\setminus T}$. 
Assume that $|\bigcup_{Q\in \mathcal{J}}Q| < \epsilon n/2$. 
Then, by the triangle inequality, $v$ is $\epsilon/2$-far from $L$.

Considering a $1$-sided test of $v$ with distribution $\mu$, we first note that no member of $\mathcal{J}$ is a witness against $v$. Thus, $\mu(\mathcal{J})$ is at most $1$ minus the probability of $\mu$ obtaining a witness. As $v$ is $\epsilon/2$-far from $L$, the probability of obtaining a witness by $\mu$ is at least $1-\delta$, implying that $\mu(\mathcal{J})<\delta$.
\end{proof}

\begin{lemma}\label{lem:LargeSupport}
For $\mu$ which is the distribution of a $1$-sided $(\epsilon/2,\delta)$-test for a non-empty property $L$ for which $\epsilon$-far words exist, let $w\in \Xi^n$ be $5\epsilon/6$-far from $L$ and $\mathcal{J}\subseteq \supp(\mu)$. If $\mu(\mathcal{J}) \geq 2\delta$, then the set $\mathcal{S}$ of all $Q\in\mathcal{J}$ which are witnesses against $w$ satisfies $|\bigcup_{Q\in \mathcal{S}}Q| \geq \epsilon n/2$.
\end{lemma}

\begin{proof}
Let $\mathcal{S}\subseteq \mathcal{J}$ be the subset of witnesses against $w$ as in the formulation of the lemma. Since $w$ is $5\epsilon/6$-far from $L$, the distribution $\mu$ provides a witness against $w$ with probability at least $1-\delta$, and therefore $\mu(\mathcal{S}) \geq \delta$. Consequently, by Lemma \ref{lem:SmallSupport}, $|\bigcup_{Q\in \mathcal{J}}Q| \geq \epsilon n/2$.
\end{proof}

\begin{definition}[$p$-sampling $(\epsilon,\delta)$-test]
A {\em $p$-sampling test} for a property $L$ is an $(\epsilon,\delta)$-test such that every $i\in [n]$ is selected as a query, independently, with probability $p$; in other words, it is a sample-based test with probability $p$ as defined in \cite{SampleBased}.
A $p$-sampling test is $1$-sided if every word in the property is accepted with probability $1$ and otherwise it is $2$-sided.
We use the notation $\mu_p$ to denote the distribution of the $p$-sampling test.
\end{definition}

\section{A conversion of a $1$-sided test to a $1$-sided sampling test}
\newcommand{\Sets}[1]{\mathcal{M}_{#1}}
\newcommand{\Core}[1]{\mathcal{C}_{#1}}
\newcommand{\Match}[1]{\mathcal{S}_{#1}}
\newcommand{\core}{core}
\newcommand{\Pom}{pompom}
\newcommand{\revealing}{revealing}

We show here that if a property is testable with a $1$-sided error, then it has $p$-sampling $1$-sided $(\epsilon,\delta)$-test with $p$ corresponding to some negative power of $n$. Specifically, we prove the following theorem, which as we explain immediately afterwards implies our claimed result.

\begin{theorem}\label{thm:1sided}
For every $n>(24q(q+1)^2(\log(|\Xi|))^2/\epsilon)^q$,
if a property over $\Xi^n$ has a $1$-sided
$(\epsilon/2,1/2(q+1),q,4(q+1)^2\log(|\Xi|)n)$-test, then it also has a
$p$-sampling $1$-sided $(\epsilon,1/2)$-test such that 
$p=O(\log(|\Xi|)q^3n^{-1/q^2}/\epsilon)$.
\end{theorem}

The preceding theorem is effective for all properties with $1$-sided $(\epsilon/2,\delta,q')$-tests, since,
by Lemma~\ref{lem:effective1sided}, 
$(\epsilon/2,\delta,q)$-tests can be converted to a $(\epsilon/2,1/2(q+1),q,4(q+1)^2\log(|\Xi|)n)$-test, where $q'$ is bounded by a polynomial in $q$ and $1/(1-\delta)$.

We next sketch a proof that the statement of Theorem~\ref{thm:1sided} holds, for every test that satisfies the additional constraint that it has a distribution $\mu$ such that $\supp(\mu)$ consists of pairwise disjoint sets.
The main result of this section can be interpreted as a reduction to this simple case.

Suppose that $\mu$ is a distribution of a  $1$-sided $(\epsilon/2,1/2(q+1),q,4(q+1)^2\log(|\Xi|)n)$-test for $L$.
Let $w$ be $\epsilon$-far from a property $L$, and $\mathcal{B}$ be the family of all the sets in $\supp(\mu)$ that are witnesses against $w$.
Now note that if $|\mathcal{B}|$ is sufficiently large, then using the fact that these sets are pairwise disjoint it is easy to show that,  
with probability at least $1/2$, the set of queries used by
a $p$-sampling test contains at least one of these sets.
This in turn implies that a $p$-sampling test reject $w$ with probability at least $1/2$. We next explain why $|\mathcal{B}|$ is indeed sufficiently large. Let $B$ be the union of all the sets in $\mathcal{B}$. 
By definition, the test rejects $w$ with probability at least $1/2$, and therefore $\mu(\mathcal{B}) \geq 1/2$.
Thus, by Lemma~\ref{lem:SmallSupport}, $|B| \geq \epsilon n/2$ and hence $|\mathcal{B}| \geq \epsilon n/(2q)$.

When the sets in $\supp(\mu)$ are not pairwise disjoint the preceding idea does not work,
 for example, in the case that the size of the intersection of all the sets in $\supp(\mu)$ is exactly of size $1$.
Here, with high probability, a set of queries selected at random according to 
$\mu_p$ does not contain a set from $\supp(\mu)$ that is a witness against $w$. Thus, we can't conclude that a $p$-sampling test rejects $w$ with the required probability.
We now explain how to circumvent this barrier in two steps: in the first step we assume that $C$, the intersection of sets in $\supp(\mu)$, is significantly smaller then $\epsilon n /2$ and 
 that we are given $w_C$ in advance; in the second step we show what to do when $w_C$ is not known in advance. 

Let $\mathcal{B}$ and $B$ be as defined in the simple case.
In the same manner as the simple case, we can conclude that 
$|B| \geq \epsilon n/2$.
Let $\mathcal{M}$ be the family of all  non-empty sets $Q\setminus C$ such that $Q\in \mathcal{\supp(\mu)}$. 
We note that the size of $\mathcal{M}$ is $O(\epsilon n/q)$ since, by construction, the size of the union of these set is $|B|-|C| = O(\epsilon n/q)$.
It is thus easy to show that, with high probability, a set of queries selected at random according to 
$\mu_p$ together with $C$ contains a witness against $w$, and hence a $p$-sampling test, with the advance knowledge of $w_C$, rejects $w$ with the required probability.
The combinatorial structure consisting of the set $C$ and the family $\mathcal{B}$ is captured by the following definition:

\begin{definition}($i$-\Pom)\label{def:Pom-pom}
A family of sets $\mathcal{S}$ is an \emph{$i$-\Pom} if there exists a set $C$, which we refer to as the {\em \core} of the $i$-\Pom, such that the following hold.
\begin{enumerate}
\item $|Q\setminus C| = i$ for every $Q\in \mathcal{S}$.
\item $Q\setminus C$ and $Q'\setminus C$ are pairwise disjoint for every distinct $Q$ and $Q'$ in $\mathcal{S}$.
\end{enumerate}
\end{definition}
The restriction of the cardinality of the sets $Q\setminus C$ is required to support technical computations in the proofs.
We next explain how the above idea works when $w_C$ is not given in advance.

If, given a word $w$, we have a large \Pom\ with core $C$ that is additionally made up of witnesses against $w_{v,C}$ for some $v\in\Xi^{|C|}$, then similarly to the simple case described above, the sampling distribution will produce a set showing that $w$ is not in the property unless $w_C\neq v$.
This simple observation is the motivation for the following setting.
Suppose that there existed a set $C$ and a set of families $\{\mathcal{S}_\sigma\}_{\sigma \in \Xi^{|C|}}$ such that
for every $\sigma \in \Xi^{|C|}$, $\mathcal{S}_\sigma$ is an $i$-\Pom\ consisting of witnesses against $w_{\sigma,C}$ and has $C$ as a \core.
Now, if the cardinality of $C$ is sufficiently small and the cardinality of every $i$-\Pom\ is sufficiently large, then we can prove the following:
with high probability,
a set of queries selected at random according to 
$\mu_p$, 
contains a set of queries whose values rule out any possible value of $w_C$, and hence imply that $w$ is not in the property.
We refer to such a set of queries as a {\em super-witness}.

\begin{definition}[super-witness against a word]
We say that $X\subseteq [n]$ is a {\em super-witness} against a word $w\in \Xi^n$, if there exists a set $Y\subseteq [n]\setminus X$ such that, for every $\sigma\in \Xi^{|Y|}$, there exists a set $Q\subseteq X\cup Y$ which is a witness against $w_{\sigma,Y}$.
\end{definition}

Recall that the set $X$ in the above definition does not necessarily contain any set from $\supp(\mu)$. However, as we prove next, it is  sufficient to imply that $w$ in not in the property. 

\begin{observation}\label{obs:SuperWitness}
Any set containing a witness against a word $w$ is also a witness against it. Additionally, a set is a super-witness against $w$ if and only if it is a witness against it.
\end{observation}
\begin{proof}
The part about containing sets follows immediately from the definition. Additionally, a witness $X$ is also a super-witness by setting $Y=\emptyset$. Now let $X\subseteq [n]$ be a super-witness against $w\in \Xi^n$. By the definition of a super-witness, for every $u\in \Xi^n$ such that $u_X = w_X$, there exists a witness against $u$ (some subset of $X\cup Y$) and hence $u\not\in L$. Thus $X$ is a witness against $w$ with regards to $L$.
\end{proof}

We now formally define the type of set of $i$-\Pom s that we need for our result.

\begin{definition}(\revealing\ set of $i$-\Pom s for $w$)\label{def:CSetofiPoms}
Let $w$ be any word in $\Xi^n$.
A set of $i$-\Pom s for $w$ is {\em \revealing} if there exists $C\subseteq [n]$ of cardinality bounded above by $4q(q+1)^2\log(|\Xi|)n^{1-i/q}$, such that,  
for every $\sigma \in \Xi^{|C|}$, the set contains an $i$-\Pom\ $\mathcal{S}_\sigma$ that satisfies:
\begin{enumerate}
\item\label{item:PomSetCore} $C$ is a \core\ of $\mathcal{S}_\sigma$.
\item\label{item:PomSetSets} for every $\sigma \in \Xi^{|C|}$,  $\mathcal{S}_\sigma$ consists only of witnesses against $w_{\sigma,C}$.
\item\label{item:PomSetCardinalities} $|\mathcal{S}_\sigma| \geq (1/3i)\epsilon\cdot n^{1-(i-1)/q}$.
\end{enumerate}
\end{definition}

The bounds on the cardinality of the \core\ and the $i$-\Pom s, follow from the combinatorial construction we use in order to prove their existence. We show next that they are sufficient for our purposes and afterwards present the combinatorial construction.

\begin{lemma}\label{lem:PompomApp}
Let $n>(24q(q+1)^2(\log(|\Xi|))^2/\epsilon)^q$, $\alpha = 15\ln{|\Xi|}\cdot q(q+1)^2/\epsilon$ and $w\in \Xi^n$.
If there exists a \revealing\ set of $i$-\Pom s for $w$ with a \core\ $C$, then the  
$\alpha \cdot n^{-1/q^2}$-sampling algorithm rejects $w$, with probability at least $\frac12$.
\end{lemma}

\begin{proof}
Let $\mathcal B_\sigma$ be the set $\{Q\setminus C:Q\in\mathcal S_\sigma\}$. We observe that $|\mathcal B_\sigma| = |\mathcal S_\sigma|$.
Let $R\subseteq [n]$ be the set of indexes sampled by the $\alpha \cdot n^{-1/q^2}$-sampling algorithm. 

Referring to Calculation \ref{calc:withi} (see appendix), for every $\sigma \in \Xi^{|C|}$, the probability that $R$ does not contain a set from  $\mathcal B_\sigma$ is at most $(1-\alpha^i\cdot n^{-i/q^2})^{(1/3i)\epsilon\cdot n^{1-(i-1)/q}} \leq \frac12 |\Xi|^{-4q(q+1)^2\log(|\Xi|)n^{1-i/q}}\leq \frac12 |\Xi|^{-|C|}$.
Thus, by the union bound, with probability at least $\frac12$, for every $\sigma \in \Xi^{|C|}$ the set $R$ contains a set from $\mathcal B_\sigma$. We note that this means that $R$ is a super-witness against $w$, which by Observation \ref{obs:SuperWitness} makes it a witness against $w$. Thus the  $\alpha \cdot n^{-1/q^2}$-sampling algorithm rejects $w$ with probability at least $\frac12$ and hence the statement of the lemma follows.
\end{proof}

We now show that, for every property $L\subseteq \Xi^n$, that has a $1$-sided $(\epsilon/2,1/2(q+1),q,4(q+1)^2\log(|\Xi|)n)$-test,
 there exists a \revealing\ set of $i$-\Pom s for $w$.
We first show that if $\supp(\mu)$ has a family of sets $\mathcal{S}$ that is almost an $i$-\Pom, then for every $w$ that is $\epsilon$-far from $L$, $\mathcal{S}$ admits a \revealing\ set of $i$-\Pom s.
By ``almost'' we mean that there exists a set $C$,
which satisfies
that every $j\in (\bigcup_{Q\in \mathcal{S}}Q)\setminus C$ is not in too many of sets in $\mathcal{S}$
(where for a true \Pom\ every such $j$ would be in exactly one set).

First, let us formally define the ``almost-pompom'' sets; this definition will also serve in the proof of the conversion for $2$-sided tests.

\begin{definition}[constellation]\label{def:constellation}
For $i\in [q]$, $n$, a distribution $\mu$ over subsets of $[n]$ of size $q$, and any positive number $\eta$, an {\em $(\eta,i)$-constellation} is a pair $(C,\mathcal S)$ consisting of a set $C\subseteq [n]$ and a family $\mathcal S\subseteq \supp(\mu)$ satisfying the following.
\begin{enumerate}
\item\label{cond:MessyCore} $|C| < \eta n^{1-i/q}$.
\item\label{cond:MessyProb} $\mu(\mathcal{S}) \geq \frac{1}{q+1}$.
\item\label{cond:MessyTails} $|Q\cap C| = q-i$, for every $Q\in\mathcal{S}$.
\item\label{cond:MessyInter} Every $j\in \bigcup_{Q\in\mathcal{S}}Q\setminus C$ is in at most $n^{(i-1)/q}$ sets from $\mathcal{S}$ if $i>1$.
\end{enumerate}
\end{definition}

\begin{lemma}\label{lem:MessyToSetOfPoms}
Let $i\in [q]$, $n > (24q(q+1)^2(\log(|\Xi|))^2/\epsilon)^q$, $w$ be any word in $\Xi^n$ that is $\epsilon$-far from $L$, and $\mu$ be the distribution of a $1$-sided $(\epsilon/2,1/2(q+1),q,4(q+1)^2\log(|\Xi|)n)$-test for $L$. If there exists a $(4q(q+1)^2\log(|\Xi|),i)$-constellation for $\mu$, then there exists a \revealing\ set of $i$-\Pom s for $w$.
\end{lemma}

\begin{proof}
Let $\sigma$ be any arbitrary word in $\Xi^{|C|}$. Since $n > (24q(q+1)^2(\log(|\Xi|))^2/\epsilon)^q$, $|C| < \epsilon n/6$, and hence by the triangle inequality $w_{\sigma,C}$ is $5\epsilon/6$-far from $L$.

Let $\mathcal S_\sigma \subseteq \mathcal S$ be the set of all $Q\in\mathcal S$ which are witnesses against $w_{\sigma,C}$. Since $\mu(\mathcal S) \geq \frac{1}{q+1}$, by Lemma~\ref{lem:LargeSupport} we have $|\bigcup_{Q\in \mathcal S_\sigma}Q| \geq \epsilon n/2$. Let $\mathcal B_\sigma$ be the set $\{Q\setminus C:Q\in\mathcal S_\sigma\}$. We observe that $|\bigcup_{Q\in\mathcal B_\sigma}Q| \geq |\bigcup_{Q\in\mathcal S_\sigma}Q| - |C| \geq \epsilon n/3$, because $|C| < \epsilon n/6$.

Suppose first that $i=1$. 
We let $\mathcal{S}_\sigma \subset \mathcal{S}$ be maximal so that, for every $Q\in \mathcal{S}_\sigma$, $Q\setminus C$ is distinct and a member of $\mathcal{B}_\sigma$.
Clearly, $|\mathcal{S}_\sigma| \geq (1/3)\epsilon\cdot n$,
$\mathcal{S}_\sigma$ is an $1$-\Pom, and $C$ is a \core\ of $\mathcal{S}_\sigma$.

Suppose now that $i>1$. Then, there exists $\mathcal B'_\sigma \subseteq \mathcal B_\sigma$ such that every pair of sets in $\mathcal{B'}_\sigma$ is disjoint and $|\mathcal B'_\sigma| \geq \frac{1}{3i}\epsilon\cdot n^{1-(i-1)/q}$, because every $j\in  \bigcup_{Q\in\mathcal B_\sigma}Q\setminus C \subseteq \bigcup_{Q\in\mathcal S}Q\setminus C$ is in at most $n^{(i-1)/q}$ sets from $\mathcal S$. 
We let $\mathcal{S}_\sigma \subset \mathcal{S}$ be maximal so that, for every $Q\in \mathcal{S}_\sigma$, $Q\setminus C$ is a distinct member of $\mathcal{B}_\sigma$.
Clearly, $|\mathcal{S}_\sigma| \geq (1/3i)\epsilon\cdot n^{1-(i-1)/q}$,
$\mathcal{S}_\sigma$ is an $i$-\Pom, and $C$ is a \core\ of $\mathcal{S}_\sigma$.

We observe that the above implies that 
for every $\sigma\in \Xi^{|C|}$,
$C$ is a \core\ of $\mathcal{S}_\sigma$,
$\mathcal{S}_\sigma$ consists only of witnesses against $w_{\sigma,C}$, and
 $|\mathcal{S}_\sigma| \geq (1/3i)\epsilon\cdot n^{1-(i-1)/q}$.
Thus, by definition,
$\{\mathcal{S}_\sigma\}_{\sigma \in \Xi^{|C|}}$ is an \revealing\ set of $i$-\Pom s for $w$.
\end{proof}

We now prove that any distribution $\mu$, over subsets of $[n]$ of size $q$, which satisfies $\supp(\mu)\leq\eta n$, admits a $(\eta,i)$-constellation $(\mathcal S,C)$ for some $i\in [q]$, as long as a certain subset of $\supp(\mu)$ (defined below) is not too heavy.
This together with Lemma~\ref{lem:PompomApp} and Lemma~\ref{lem:MessyToSetOfPoms}, substituting $\eta=4q(q+1)^2\log(|\Xi|)$ and referring to the distribution $\mu$ of the $(\epsilon/2,1/2(q+1),q,4(q+1)^2\log(|\Xi|) n)$-test, is sufficient for the proof of Theorem~\ref{thm:1sided}.

We start by defining three sets of families, 
$\{\Sets{i}\}_{i=0}^{q}$, $\{\Core{i}\}_{i=0}^{q}$ and $\{\Match{i}\}_{i=0}^{q}$, where $\{\Sets{i}\}_{i=0}^{q}$ is 
 a partition of $\supp(\mu)$.
 We prove afterwards that, as long as $\mu(\mathcal S_0)\leq\frac1{q+1}$, for some $i\in [q]$
 the sets $\Sets{i}$ and $\Core{i}$ respectively compose the claimed constellation $(\mathcal S,C)$.

\begin{definition}[$\Sets{i}$, $\Core{i}$ and $\Match{i}$]\label{def:SCM}
Given a distribution $\mu$ over subsets of $[n]$ of size $q$ whose support is bounded by $\eta n$, we inductively define $\Sets{i}$, $\Core{i}$ and $\Match{i}$ as follows.
\begin{enumerate}
\item Let $\Sets{0}=\supp(\mu)$, and $\Core{0}$ be the set of indexes $j\in [n]$ such that $j$ is a member of at least $n^{q^{-1}}$ sets in $\Sets{0}$.\label{item:defiSCMc0}
\item For $i=0,1,\dots,q$, after $\Sets{i}$ and $\Core{i}$ are defined, let $\Match{i}$ be all the sets $Q\in \Sets{i}$ such that $|Q\cap \Core{i}| = q-i$.\label{item:defiSCMsi}
\item For $i=1,\dots,q$, after $\Sets{i}$ is defined, let $\Core{i}$ be the set of indexes $j\in [n]$ such that $j$ is a member of at least $n^{i/q}$ sets in $\Sets{i}$.\label{item:defiSCMci}
\item For $i=1,\dots,q$, after $\Sets{i-1}$ and $\Match{i-1}$ are defined, let $\Sets{i} = \Sets{i-1}\setminus \Match{i-1}$.\label{item:defiSCMmi}
\end{enumerate}
\end{definition}

The following statements give the properties of these sets.

\begin{observation}\label{obs:SCMcombi}
The following hold for the sets of Definition \ref{def:SCM} when they are constructed from a distribution $\mu$ satisfying the conditions there.
\begin{enumerate}
\item\label{item:CCoreSize0} $|\Core{0}| < \eta qn^{1-1/q}$.
\item\label{item:CCoreSize} $|\Core{i}| < \eta qn^{1-i/q}$ for all $1\leq i\leq q$.
\item\label{item:CSetsContained} $\Sets{i} \subseteq \Sets{i-1}$ for all $1\leq i\leq q$.
\item\label{item:CCoreContained} $\Core{i}  \subseteq \Core{i-1}$ for all $1\leq i\leq q$.
\item\label{item:CMAtch} $\Match{i}\cap \Match{j}= \emptyset$ for all $1\leq i< j\leq q$.
\end{enumerate}
\end{observation}

\begin{proof}
Item \ref{item:CSetsContained} follows immediately from the construction, and implies Item \ref{item:CCoreContained} in turn. Item \ref{item:CCoreSize0} follows from the definition along with the assumption on the support size of $\mu$, and so does Item \ref{item:CCoreSize} using Item \ref{item:CSetsContained}. For Item \ref{item:CMAtch} assume that $i<j$ and note the construction of $\Sets{i+1}$, which makes it disjoint from $\Match{i}$ while containing $\Match{j}$.
\end{proof}

The goal of the following two lemmas is to prove that, for every $i \in [q]$ and $Q\in \Match{i}$, we have that $(Q\setminus \Core{i})\cap \Core{i-1} = \emptyset$.
According to Definition~\ref{def:SCM}, this implies that 
$\Sets{i}$ and $\Match{i}$ satisfy Condition~\ref{cond:MessyInter} of Lemma~\ref{lem:MessyToSetOfPoms}.
At a high level of abstraction the proof starts with the assumption for some $i\in [q]$ there exists $Q\in \Match{i}$ such that $(Q\setminus \Core{i})\cap \Core{i-1} \neq \emptyset$; afterwards it is shown that this $Q$ is in $\Sets{j}$, for some $j\in [i-1]$; this by Definition~\ref{def:SCM} implies that $Q\not\in \Match{i}$ in contradiction to the assumption that $Q\in \Match{i}$.
The following lemma is used in order to restrict the setting to that where $j = i-1$.

\begin{lemma}\label{lem:SCMmatchcap}
For  $i = 0,1,\dots,q$ and every $Q\in \Sets{i}$ we have that $|Q\cap \Core{i}| \leq q-i$.
\end{lemma}

\begin{proof}
By definition, $|Q| \leq q$ for every $Q\in \Sets{0}$. Hence, $|Q\cap \Core{0}| \leq q-0  =q$. We proceed by induction over $i$. Assume that the statement of the lemma holds for $i-1\geq 0$. Suppose for the sake of contradiction that there exists $Q \in \Sets{i}$ such that $|Q\cap \Core{i}| > q-i$. 
Since $\Core{i} \subseteq \Core{i-1}$, by Item \ref{item:CCoreContained} of Observation \ref{obs:SCMcombi}, 
this implies that $|Q\cap \Core{i-1}| \geq q-(i-1)$. 
Hence, $|Q\cap \Core{i-1}| = q-(i-1)$, by the induction assumption. 
Therefore, $Q\in \Match{i-1}$, because, by construction, we also have that $Q \in \Sets{i-1}$. 
Consequently, we get the contradiction that $Q\not\in \Sets{i}$, 
since $\Sets{i} = \Sets{i-1}\setminus \Match{i-1}$.
\end{proof}

\begin{lemma}\label{lem:SCMmatchout}
For  $i = 1,\dots,q$ and every $Q\in \Match{i}$ we have that $(Q\setminus \Core{i})\cap \Core{i-1} = \emptyset$.
\end{lemma}

\begin{proof}
We proceed by induction over $i$. The base case is $i=0$ which follows from the definition of $\Match{0}$, even if we set $\Core{-1}=[n]$. Assume that the statement of the lemma holds for $i-1$. Suppose for the sake of contradiction that there exists $Q \in \Match{i}$ such that $|(Q\setminus \Core{i})\cap \Core{i-1}| > 0$. Thus, $|Q\cap \Core{i-1}|=|(Q\setminus \Core{i})\cap \Core{i-1}|+|Q\cap \Core{i}| \geq q-(i-1)$, because $\Core{i} \subseteq \Core{i-1}$, by Item \ref{item:CCoreContained} of Observation \ref{obs:SCMcombi}. 
Therefore, by Lemma \ref{lem:SCMmatchcap}, $|Q\cap \Core{i-1}| = q-(i-1)$. 
Since by construction we also have that $Q \in \Sets{i-1}$ we deduce that 
 $Q\in \Match{i-1}$. 
 Consequently, we get the contradiction that $Q\not\in \Match{i}$, since $\Match{i}\subseteq \Sets{i} = \Sets{i-1}\setminus \Match{i-1}$.
\end{proof}

These last lemmas imply that if $\mu(\mathcal S_0)$ is not large, then a constellation exists.

\begin{lemma}\label{lem:SCMconstellation}
If $\mu(\mathcal S_0)\leq\frac1{q+1}$ then for some $i\in [q]$ the pair $(\Match{i},\Core{i})$ is an $(\eta q,i)$-constellation for $\mu$.
\end{lemma}

\begin{proof}
By the assumption $\mu(\Match{0}) \leq 1/(q+1)$, averaging and Item \ref{item:CMAtch} of Observation \ref{obs:SCMcombi}, there exists $i\in [q]$ such that $\mu(\Match{i}) \geq 1/(q+1)$. Consequently, following from Observation \ref{obs:SCMcombi} Items \ref{item:CCoreSize0},\ref{item:CCoreSize}, the choice of $i$, Definition \ref{def:SCM} Item \ref{item:defiSCMsi}, and Lemma \ref{lem:SCMmatchout} together with Definition \ref{def:SCM} Items \ref{item:defiSCMc0},\ref{item:defiSCMci}, in that order, $\Core{i}$ and $\Match{i}$ satisfy the four conditions of Definition \ref{def:constellation} and thus form an $(\eta q,i)$-constellation.
\end{proof}

We now prove the $1$-sided test conversion result.

\begin{proof}[Proof of Theorem~\ref{thm:1sided}]
Given a distribution $\mu$ which corresponds to a $1$-sided $(\epsilon/2,1/2(q+1),q,4(q+1)^2\log(|\Xi|)n)$-test for a property over $\Xi^n$, where $n>(24q(q+1)^2(\log(|\Xi|))^2/\epsilon)^q$, we show for $\alpha=15\ln{|\Xi|}\cdot q(q+1)^2/\epsilon$ that the $\alpha\cdot n^{-1/q^2}$-sampling distribution corresponds to a $1$-sided $(\epsilon,1/2)$-test for the same property.

From definition \ref{def:SCM}, every set in $\Match{0}$ is in particular a subset of $\Core{0}$. Hence, the union of all the sets in $\Match{0}$ is also a subset of $\Core{0}$ and therefore $|\bigcup_{Q\in \Match{0}}Q| \leq |\Core{0}|$. Thus, by Item \ref{item:CCoreSize0} of Observation \ref{obs:SCMcombi}, $|\bigcup_{Q\in \Match{0}}Q| \leq |\Core{0}| < 4q(q+1)^2\log(|\Xi|) n^{1-1/q} < \epsilon n /2$, where the last inequality follows from  $n > (24q(q+1)^2(\log(|\Xi|))^2/\epsilon)^q$. Therefore, by Lemma \ref{lem:SmallSupport}, $\mu(\Match{0}) \leq 1/(q+1)$, and so by Lemma \ref{lem:SCMconstellation} there exists an $(4q(q+1)^2\log(|\Xi|),i)$-constellation for some $i\in Q$.
Therefore, by Lemma \ref{lem:MessyToSetOfPoms}, there exists a \revealing\ set of $i$-\Pom s for $w$.
Consequently, by Lemma~\ref{lem:PompomApp} the theorem follows.
\end{proof}

\section{Probabilistic formulas and test combinatorialization}
Here we take a non-adaptive $2$-sided test and make its structure more malleable to combinatorial arguments, with the main feature being that the new query distribution will be uniform over its support. At first, we define a structure that can generally describe tests; we use this formulation to make the following arguments clearer and more succinct, which will also present them in their fullest possible generality.

\begin{definition}[probabilistic constraints and formulas]\label{def:probfor}
A {\em probabilistic $q$-constraint} (over an alphabet $\Xi$) is a pair $C=(Q,S)$ where $Q\subseteq [n]$ is a {\em constraint set}, also called a {\em query set}, of size $q$, and $S$ is a {\em satisfaction function} from $\Xi^{|Q|}$ to the real interval $[0,1]$.

A {\em probabilistic $q$-formula} $P=(\mathcal F,\mu)$ is a set $\mathcal F$ of $q$-constraints, all with distinct constraint sets, along with a probability distribution $\mu$ over $\mathcal F$. We call it a $(q,k)$-formula if additionally $|\supp(\mu)|\leq k$, in which case we can assume that $|\mathcal F|\leq k$.

When we drop the restriction on the sizes of the query sets of the constraints (even the restriction that they are all of the same size) then we call $P$ a {\em probabilistic formula}.

Given a word $w\in\Xi^n$ and a probabilistic formula $P$, the {\em satisfaction} of $P$ by $w$ is the average of the random variable that results from picking a constraint $(Q,S)\in \mathcal F$ according to $\mu$ and obtaining the value $S(w_Q)$. $P$ is said to be {\em $\delta$-sure for $w$} if its satisfaction by $w$ is either at least $1-\delta$ or at most $\delta$.
\end{definition}

The requirement for all sets corresponding to constraints being distinct allows us (given a particular formula $\mathcal F$) to identify the distribution $\mu$ with the corresponding distribution over subsets of $[n]$ only. This we will do throughout the sequel, but first let us justify this requirement.

\begin{lemma}\label{lem:unionize}
The requirement that the members of $P$ have distinct query sets is without loss of generality.
\end{lemma}

\begin{proof}
If $C_1=(Q,S_1)$ and $C_2=(Q,S_2)$ are two constraints in a formula $P=(\mathcal F,\mu)$ (that for now does not satisfy the distinct set requirement), then we define $\mathcal F'$ by replacing them with $C=(Q,S)$ where $S=(\mu(C_1)\cdot S_1+\mu(C_2)\cdot S_2)/(\mu(C_1)+\mu(C_2))$, and define the corresponding $\mu'$ by setting $\mu'(C)=\mu(C_1)+\mu(C_2)$. This preserves satisfaction values over all words $w$. We can continue doing this until there are no pairs left of constraints sharing the same query set.
\end{proof}

We shall henceforth abuse notation, and indeed refer to $\mu$ both as a distribution over $2^{[n]}$ and as a distribution over $\mathcal F$. Also, we shall make liberal use of the assumption (without loss of generality) that the support of $\mu$ is the entire $\mathcal F$ (otherwise we replace it with the appropriate subset).

A non-adaptive $2$-sided test or partial test can be described as follows.

\begin{definition}[alternative definition of non-adaptive tests]
Given two properties $L'\subseteq L\subseteq \Xi^n$, a {\em non-adaptive} $2$-sided partial $(\epsilon,\delta)$-test for $(L',L)$ is a probabilistic formula whose satisfaction over any $w\in L'$ is at least $1-\delta$, while its satisfaction for any $w\in \Xi^n$ that is $\epsilon$-far from $L$ is at most $\delta$.

If $L'=L$ then we just call it a $2$-sided $(\epsilon,\delta)$-test for $L$.

If the test uses a $q$-formula then we may also call it an $(\epsilon,\delta,q)$-test, and if it uses a $(q,k)$-formula then we may call it an $(\epsilon,\delta,q,k)$-test.
\end{definition}

To convert a non-adaptive test to this definition, we take $\mu$ to be the query distribution corresponding to the test, and set each pair $(Q,S)$ so that $S$ will describe the acceptance probability of the test given each possible outcome of its queries to $Q$.

We need the following technicality for pairs $(L',L)$. It is safe to restrict our discussion to such pairs because otherwise there exists a corresponding trivial partial test.

\begin{definition}
Given two properties $L'\subseteq L\subseteq \Xi^n$, we say that the pair $(L',L)$ is {\em $\epsilon$-nontrivial} if there exist some word in $L'$ and some word $\epsilon$-far from $L$.
\end{definition}

The main purpose of this section is to show that all tests can be made to obey certain restrictions, at some reasonable cost for their parameters. To formulate the main lemma we need to define what these restrictions may be.

\begin{definition}[restrictions on formulas and tests]
A probabilistic formula $P$ is said to be {\em zero-one} if all its constraints have the range $\{0,1\}$ (instead of the whole interval).

$P$ is said to be {\em $\beta$-equitable} if for every two constraints $C_1$ and $C_2$ in the support of the corresponding distribution $\mu$, we have $\mu(C_1)\leq \beta\mu(C_2)$. In particular, for a $1$-equitable formula the distribution $\mu$ is uniform over its support.

A $q$-formula $P$ is said to be {\em combinatorial} if it is a zero-one and equitable.

We use the same adjectives for tests. For example a test is called {\em combinatorial} if its corresponding formula is combinatorial.
\end{definition}

We will prove the main combinatorialization lemma of this section following a sequence of steps. The easiest of these steps is making the corresponding formula zero-one.

\begin{lemma}\label{lem:zeroone}
A formula $P$ can be made into a a zero-one formula $P'$ without any change in its other parameters (including also support size and equitability), so that for any input for which $P$ was $\delta$-sure about, $P'$ will be $2\delta$-sure about and in the same direction.
\end{lemma}

\begin{proof}
For every constraint $C=(Q,S)$ in $\supp(\mu)$, we replace it with $C'=(Q,S')$, where $S'$ is defined such that $S'(v)=0$ if $S(v)<\frac12$, and otherwise $S'(v)=1$. We leave $\mu$ ``unmodified'', that is the new $\mu'$ is defined by having $\mu'(C')=\mu(C)$, in particular remaining identical as a distribution over query sets.

We present here the analysis for the case where the satisfaction of $P$ by $w\in\Xi^n$ is at most $\delta$. The case where it is at least $1-\delta$ is symmetric. Given such a $w$, we set $\mathcal F=\supp(\mu)$, and let $\mathcal F_2$ be the set of clauses whose satisfaction by $w$ is at least $1/2$. Clearly $\mu(\mathcal F_2)\leq 2\delta$. The satisfaction of $P'$ by $w$ is now bounded by $0\cdot\mu(\mathcal F\setminus\mathcal F_2)+1\cdot\mu(\mathcal F_2)\leq 2\delta$.
\end{proof}

In the sequel we will need to analyze formulas conditioned on subsets of the original constraint set.

\begin{definition}
Given a probabilistic formula $P=(\mathcal F,\mu)$ and $\emptyset\neq \mathcal F'\subseteq \mathcal F$, the {\em $\mathcal F'$-conditioned} formula is $P'=(\mathcal F',\mu')$ where $\mu'$ is $\mu$ conditioned on the event that a member from $\mathcal F'$ was chosen.
\end{definition}

The following fact about conditioned formulas is trivial.

\begin{observation}\label{obs:condsure}
Given $P=(\mathcal F,\mu)$, if $\mathcal F'\subseteq \mathcal F$ satisfies $\mu(\mathcal F')\geq\eta$, then for every input for which $P$ was $\delta$-sure about, the conditioned formula $P'$ will be $\delta/\eta$-sure about and in the same direction.
\end{observation}

\begin{proof}
Again we analyze the case where the satisfaction of $P$ by $w\in\Xi^n$ is at most $\delta$, as the case where it is at least $1-\delta$ is symmetric. For such $w$ we write:
$$\sum_{(Q,S)\in\mathcal F'}\mu'(Q)S(w_Q)=\sum_{(Q,S)\in\mathcal F'}\mu(Q)S(w_Q)/\mu(\mathcal F')\leq \delta/\eta$$
where in the symmetric case we refer to $(1-S(w_Q))$ instead of $S(w_Q)$.
\end{proof}

We next prove a lemma (which like most preceding ones holds also for formulas which are not tests), that allows us to move from $\beta$-equitable formulas all the way to $1$-equitable ones. For making the transition cost not too high, we prove first a ``quantization'' step.

\begin{lemma}\label{lem:equiquant}
A $\beta$-equitable formula $P=(\mathcal F,\mu)$ can be made into a formula $P'=(\mathcal F,\mu')$ for which $\mu'$ has at most $\log(2\beta)$ possible values, so that for any input for which $P$ was $\delta$-sure about, $P'$ will be $2\delta$-sure about and in the same direction. Moreover, since $P'$ has the same $\mathcal F$ and the same support, it preserves the original support size, query size, and zero-one property (if it existed) of $P$.
\end{lemma}

\begin{proof}
We first define $\tilde{\mu}$ by setting for every $C\in\mathcal F$ the value $\tilde{\mu}(C)$ to be $2^{-k_C}$, where $k_C$ is the largest integer for which $2^{-k_C}\geq\mu(C)$. Clearly for every $C$ we have $\mu(C)\leq\tilde{\mu}(C)\leq 2\mu(C)$, and clearly $\tilde{\mu}$ has at most $\log(2\beta)$ possible values. However, it is not a probability measure, because it may be that $\tilde{\mu}(\mathcal F)>1$. We thus set $\mu'(C)=\tilde{\mu}(C)/\tilde{\mu}(\mathcal F)$ for every $C\in\mathcal F$.

Finally, if the satisfaction of $P$ by $w$ is at most $\delta$, we write:
$$\sum_{(Q,S)\in\mathcal F}\mu'(Q)S(w_Q)\leq \sum_{(Q,S)\in\mathcal F}\tilde{\mu}(Q)S(w_Q)\leq 2\cdot\!\!\!\!\sum_{(Q,S)\in\mathcal F}\mu(Q)S(w_Q)\leq 2\delta$$
where again the case of the satisfaction being at least $1-\delta$ is symmetric.
\end{proof}

\begin{lemma}\label{lem:equifin}
A $\beta$-equitable formula $P=(\mathcal F,\mu)$ can be made into a $1$-equitable formula $P'=(\mathcal F',\mu')$, so that for any input for which $P$ was $\delta$-sure about, $P'$ will be $2\delta\log(2\beta)$-sure about and in the same direction. Moreover, $\mathcal F'\subseteq \mathcal F$, so $P'$ preserves the support size bound, query size bound, and possible zero-one property of $P$.
\end{lemma}

\begin{proof}
We first use Lemma \ref{lem:equiquant} to move from $P$ to $P''=(\mathcal F,\mu'')$, where $\mu''$ has at most $\log(2\beta)$ possible values, and any input for which $P$ was $\delta$-sure about, $P''$ is $2\delta$-sure about. Now there must be some $\eta\in (0,1]$ so that $\mathcal F_{\eta}=\{C\in\mathcal F:\mu''(C)=\eta\}$ satisfies $\mu''(F_{\eta})\geq 1/\log(2\beta)$.

We set $\mathcal F'=\mathcal F_{\eta}$ and make $P'$ the formula of $P''$ conditioned on $\mathcal F'$. We finalize the proof by appealing to Observation \ref{obs:condsure}.
\end{proof}

Specifically for (partial) tests, we next show some correlation between subsets ``covering'' few indexes and probability. The following lemma will also be used again when constructing sets of pompoms to prove the main conversion result from $2$-sided tests to sampling tests.

\begin{lemma}\label{lem:lowprob}
If a formula $P=(\mathcal F,\mu)$ corresponds to an $(\epsilon/2,\delta)$-test for $(L',L)$ which is $\epsilon$-nontrivial, and $\mathcal F'\subseteq\mathcal F$ is such that the union of its corresponding query sets occupies at most $\epsilon n/2$ indexes from $[n]$, then $\mu(\mathcal F')\leq 2\delta$.
\end{lemma}

\begin{proof}
Let $T = \bigcup_{(Q,S)\in\mathcal F'}Q$, $u\in\Xi^n$ be a word in $L'$, $w\in \Xi^n$ be $\epsilon$-far from $L$, and $v$ be such that $v_T = u_T$ and $v_{[n]\setminus T} = w_{[n]\setminus T}$. 
By the triangle inequality, $v$ is $\epsilon/2$-far from $L$, and so
$$\sum_{(Q,S)\in\mathcal F'}\mu(Q)S(u_Q) = \sum_{(Q,S)\in\mathcal F'}\mu(Q)S(v_Q) \leq \delta$$
since this bounds the satisfaction of $P$ by $v$.

On the other hand, the satisfaction of $P$ by $u$ is at least $1-\delta$, and so we obtain
$$1-\delta \leq \sum_{(Q,S)\in\mathcal F'}\mu(Q)S(u_Q)+\sum_{(Q,S)\in\mathcal F\setminus\mathcal F'}\mu(Q)S(u_Q) \leq \delta+(1-\mu(\mathcal F'))$$
which necessarily means that $\mu(\mathcal F')\leq 2\delta$.
\end{proof}

Using the above lemma we can show that tests with a small enough support size can be made into equitable ones.

\begin{lemma}\label{lem:linequi}
For $\delta<\frac18$, an $(\epsilon/2,\delta,q,\alpha n)$-test for $(L',L)$ which is $\epsilon$-nontrivial can be made into a $\beta$-equitable $(\epsilon/2,2\delta,q,\alpha n)$-test for $(L',L)$, for $\beta=8q\alpha/\epsilon$. This transformation also preserves the zero-one property if it existed.
\end{lemma}

\begin{proof}
Let $P=(\mathcal F,\mu)$ be the formula corresponding to the test. First we let $\mathcal F_0=\{C\in\mathcal F:\mu(C)\leq 1/4\alpha n\}$. Clearly $\mu(\mathcal F_0)\leq\frac14$. Now let $\mathcal F_1=\{C\in\mathcal F:\mu(C)\geq 2q/\epsilon n\}$. Clearly $|\mathcal F_1|\leq \epsilon n/2q$. Hence $|\bigcup_{(Q,S)\in\mathcal F_1}Q|\leq\epsilon n/2$, and so by Lemma \ref{lem:lowprob} we have $\mu(\mathcal F_1)\leq 2\delta\leq\frac14$.

Setting $\mathcal F'=\mathcal F\setminus (\mathcal F_0\cup\mathcal F_1)$, we get $\mu(\mathcal F')\geq \frac12$. Setting $P'$ to be the conditioning of $P$ to $\mathcal F'$, we obtain by Observation \ref{obs:condsure} that it is an $(\epsilon/2,2\delta,q,\alpha n)$-test for $(L',L)$. Moreover, since for every $C\in\mathcal F'$ we have $1/4\alpha n<\mu(C)<2q/\epsilon n$, we get that the resulting test is $\beta$-equitable for $\beta=8q\alpha/\epsilon$.
\end{proof}

We now have nearly all the ingredients we need. The final one is a is a way to convert a general test to one whose support size is linear in $n$, which the following lemma provides even for formulas that are not necessarily tests.

\begin{lemma}\label{lem:linearize}
Any $q$-formula $P=(\mathcal F,\mu)$ can be made into a $(q,\alpha n)$-formula $P'$ for $\alpha=\delta^{-2}\log(|\Xi|)$, with the condition that any $w\in\Xi^n$, for which $P$ was $\delta$-sure about, $P'$ will be $2\delta$-sure about and in the same direction. This also preserves the zero-one property if it exists.
\end{lemma}

\begin{proof}
To produce the new formula, we take $r=\delta^{-2}\log(|\Xi|)\cdot n$ samples $(Q_1,S_1),\ldots,(Q_r,S_r)$ from $\mathcal F$ by independently drawing each sample according to $\mu$. For $w\in\Xi^n$ we set $\eta_w=(\sum_{i=1}^rS_i(w_{Q_i}))/r$. We also let $\eta=\sum_{(Q,S)\in\mathcal{F}}\mu(Q)S(w_Q)$ denote the satisfaction of $P$ by $w$.
Let $Y_i$ denote the random variable $S_i(w_{Q_i})$, and set $Y=\sum_{i=1}^rY_i/r$. Note that $\mathrm E[Y_i]=\sum_{(Q,S)\in\mathcal{F}}\mu(Q)S(w_Q)=\eta$. Thus also $\mathrm E[Y]=\eta$, and by Lemma \ref{lem:hoeff-bound} we have that the probability for $|\eta_w-\eta|>\delta$ is bounded by $2e^{-2r\delta^2}\leq \frac12|\Xi|^{-n}$.

Thus, with probability at least $\frac12$, the obtained sequence is such that for all $w\in\Xi^n$ we have that the difference between $\eta_w$ and $\eta$ is at most $\delta$. We fix such a sequence $(Q_1,S_1),\ldots,(Q_r,S_r)$. To define $P'=(\mathcal F',\mu')$, we set $\mathcal F'$ to be the set of clauses appearing in $(Q_1,S_1),\ldots,(Q_r,S_r)$, where for $C\in\mathcal F'$ we set $\mu'(C)$ to be the number of times it appeared in the sequence, divided by $r$.
\end{proof}

Now we are finally ready to prove the main combinatorialization result of this section.

\begin{lemma}[combinatorialization lemma]\label{lem:combi}
Any (partial) $(\epsilon/2,\delta,q)$-test for $(L',L)$ which is $\epsilon$-nontrivial can be made into a combinatorial $(\epsilon/2,\delta',q,\alpha n)$-test, where $\alpha=\delta^{-2}\log(|\Xi|)$ and $\delta'=16\delta\log(16q\delta^{-2}\log(|\Xi|)/\epsilon)$.
\end{lemma}

\begin{proof}
Setting $P$ to be the formula corresponding to the test, we can assume that $\delta<\frac1{16}$, as otherwise we can just ignore $P$ and provide a ``test'' that is satisfied by all inputs. We perform the following sequence of steps.

\begin{itemize}
\item Use Lemma \ref{lem:linearize} to make it into a formula $P_1$ corresponding to an $(\epsilon/2,2\delta,q,\alpha n)$-test for $\alpha=\delta^{-2}\log(|\Xi|)$.
\item Use Lemma \ref{lem:linequi} to make $P_1$ into a formula $P_2$ that is an $(8q\delta^{-2}\log(|\Xi|)/\epsilon)$-equitable $(\epsilon/2,4\delta,q,\alpha n)$-test. This is the only step that requires the formula to correspond to a test.
\item Use Lemma \ref{lem:equifin} to make $P_2$ into $P_3$ which is an $(\epsilon/2,8\delta\log(16q\delta^{-2}\log(|\Xi|)/\epsilon),q,\alpha n)$-test that is $1$-equitable.
\item Finally use Lemma \ref{lem:zeroone} to make $P_3$ into the formula $P'$, which is a combinatorial partial $(\epsilon/2,16\delta\log(16q\delta^{-2}\log(|\Xi|)/\epsilon),q,\alpha n)$-test for $(L',L)$.
\end{itemize}

The final formula $P'$ is the required test.
\end{proof}

Before concluding this section, we note that by analyzing the only place in the proof where we used that $P$ is a test, we can formulate the following more general ``combinatorialization lemma'' which could be of independent use.

\begin{lemma}[general combinatorialization]
If $P$ is a probabilistic $q$-formula resolving with $\delta$ confidence, for which there are both ``yes'' instances, and words where every $\epsilon/2$-close word is a ``no'' instance, can be made into a combinatorial $(q,\alpha n)$ formula resolving the same problem with $\delta'$ confidence, where $\alpha$ and $\delta'$ are as in Lemma \ref{lem:combi}.
\end{lemma}

To conclude this section, we combine Lemma \ref{lem:combi} with an amplification technique to show how general $2$-sided $(\epsilon/2,\delta,q)$-tests can be converted to combinatorial tests.

\begin{lemma}\label{lem:effective2sided}
A partial $(\epsilon/2,\delta,q)$-test for $(L',L)$ which is $\epsilon$-nontrivial can be converted to a combinatorial partial $(\epsilon/2,1/(4q')^3,q',q'^9\log(|\Xi|)(\log(\log(|\Xi|)/\epsilon))^2 n)$-test for $(L',L)$, where $q'=O(q\log(q)\log\log(\log(|\Xi|)/\epsilon)/(\frac12-\delta)^2)$.
\end{lemma}

\begin{proof}
We first use $2$-sided amplification for the original test: We repeat the original test $20\log(q)\log\log(\log(|\Xi|)/\epsilon)/(\frac12-\delta)^2$ times and take the majority vote. This brings us to an $(\epsilon/2,1/40q'^4\log(\log(|\Xi|)/\epsilon),q')$-test for $q'=O(q\log(q)\log\log(\log(|\Xi|)/\epsilon)/(\frac12-\delta)^2)$. We use Lemma \ref{lem:combi}, and obtain from it a combinatorial $(\epsilon/2,1/4q'^3,q',q'^9\log(|\Xi|)(\log(\log(|\Xi|)/\epsilon))^2 n)$-test as required for the lemma's conclusion.
\end{proof}

We note that the dependency of $q'$ above on $\log\log(\log(|\Xi|)/\epsilon)$ implies that a constant power of $n$ is guaranteed only for properties for which the alphabet does not depend on $n$. However, it is unlikely to destroy sublinearity by itself even for a variable alphabet setting, because already for $|\Xi|=2^n$ there are examples with no sublinear sampling tests at all by \cite{SampleBased}.

\section{A conversion of a $2$-sided test to a $2$-sided sampling test}
\newcommand{\good}{discerning}

\newcommand{\alphaValue}{10^3\ln(|\Xi|)\cdot q^4/\epsilon}

\newcommand{\TwoSidednLB}{(24q^{10}(\log(|\Xi|))^2/\epsilon)^q}

\newcommand{\twoSidedPomCard}{\epsilon\cdot n^{1-(i-1)/q}/(3i)}

Here we prove that if the properties $L'\subseteq L\subseteq\Xi^n$ admit a $2$-sided test with a constant number of queries for $(L',L)$, then there is a corresponding $2$-sided $p$-sampling test where $p$ corresponds to a constant negative power of $n$. Specifically we prove the following.

\begin{theorem}\label{thm:2sided}
Let $\alpha = \alphaValue$. For every $q\geq 3$ and $n>\TwoSidednLB$, if $(L',L)$ over $\Xi^n$ admits a $2$-sided combinatorial $(\epsilon/2,1/4q^3,q,q^9\log(|\Xi|)(\log(\log(|\Xi|)/\epsilon))^2n)$-test and is $\epsilon$-nontrivial,
then it also admits a $p$-sampling $2$-sided $(\epsilon,1/10)$-test such that $p=\alpha n^{-1/q^2}$.
\end{theorem}

As with the proof of the $1$-sided case, we set $\mu$ to be a distribution of the test, and find for it pompoms that cover every possible assignment to a common core set $C$. Here however there are many pompoms involved,  and they serve all assignments to $C$ at once, because we need to cover enough of the ``weight'' of the distribution $\mu$. Also, the pompoms are not necessarily of witnesses, but rather of query sets; we will use them to approximate for every assignment of $C$ the ``amount'' of query sets that would cause rejection (hence they need to cover sufficient weight). We need the test to be combinatorial, i.e., that $\mu$ is uniform over its support, in exactly one place: The sampling test will approximate the {\em number} of rejecting query sets, and only for a uniform $\mu$ will this correspond to the rejection {\em probability} of the original test.

Let us formally define the set of pompoms that we will use.

\begin{definition}\label{def:good}
Given a distribution $\mu$ over sets of size $q$,
a set $J$ of $i$-\Pom s made from members of $\supp(\mu)$ is {\em \good} for $\mu$
if the following holds:
\begin{enumerate}
\item\label{item:GoodI} $\mu(\mathcal I) \geq \frac{1}{2(q+1)}$, where $\mathcal I=\bigcup_{\mathcal W\in J}\mathcal W$ is the union of all the 
$i$-\Pom s in $J$.
\item\label{item:GoodPomSize} Every $i$-\Pom\ in $J$ has cardinality exactly $\twoSidedPomCard$.
\item\label{item:GoodCoreSize} There exists $C\subseteq [n]$ of size at most $q^{10}\log(|\Xi|)(\log(\log(|\Xi|)/\epsilon))^2n^{1-i/q}$
that is a \core\ of all the $i$-\Pom s in $J$.
\end{enumerate}
\end{definition}

Next we define and show how each pompom of such a set can be used for the proof of something like Theorem \ref{thm:2sided}.

\begin{definition}\label{def:PompomApprox}
Given a probabilistic $q$-formula $P=(\mathcal F,\mu)$ over $\Xi^n$, an $i$-pompom $\mathcal W\subseteq\supp(\mu)$ of size $\twoSidedPomCard$ with core $C$ of size at most $q^{10}\log(|\Xi|)(\log(\log(|\Xi|)/\epsilon))^2n^{1-i/q}$, a word $w\in\Xi^n$, a possible assignment $\sigma\in\Xi^{|C|}$ to $C$, and a query set $U\subseteq [n]$, the {\em approximated satisfiability in $\mathcal W$ of $\sigma$ with respect to $w$} is defined to be the value $\gamma_{\sigma,U,\mathcal W}$ obtained in the following manner.

Set $\mathcal W_U=\{Q\in \mathcal W:Q\setminus C\subseteq U\}$ (i.e., take the set of members of $\mathcal W$ whose indexes outside $C$ are contained in $U$), and then take the average
$\gamma_{\sigma,U,\mathcal W}=(\sum_{\{(Q,S)\in\mathcal F:Q\in\mathcal W_U\}}S((w_{\sigma,C})_Q))/|\mathcal W_U|$,
which we arbitrarily set to $\frac12$ if $\mathcal W_U=\emptyset$.
\end{definition}

An explanation to the above definition: Assume that $U$ is a set of queries that we have made. We would like to assess the assignment $\sigma$ of $C$, with respect to what $U$ tells us about $w$ outside of $C$. Given the $i$-pompom $\mathcal W$, we want to approximate the relative weight of the members of $\mathcal W$ for which the corresponding constraints accept $w_{\sigma,C}$. We do so by restricting ourselves to the members of $\mathcal W_U$, for which we can tell by querying $U$ whether they accept $w_{\sigma,C}$ or not. We ignore all aspects of $\mu$ apart from its support, because we will assume that it is uniform over $\supp(\mu)$ (i.e., that the formula $P$ corresponds to a combinatorial test). This assumption is essential to show that a set $U$ chosen according to a sampling distribution will indeed yield with high probability a good approximation.

Note that $\gamma_{\sigma,[n],\mathcal W}$ is the true acceptance average of the pompom $\mathcal W$. We now prove that the sampling distribution with high probability provides a $U$ such that $\gamma_{\sigma,U,\mathcal W}$ approximates $\gamma_{\sigma,[n],\mathcal W}$.

\begin{lemma}\label{lem:approxone}
Let $q\geq 3$, $n>\TwoSidednLB$, $\alpha = \alphaValue$ and $w\in \Xi^n$. Suppose that the formula $P=(\mathcal F,\mu)$, the $i$-pompom $\mathcal W$ and its core $C$, and the words $w$ and $\sigma$ are as per the requirements of Definition \ref{def:PompomApprox}, and additionally that $P$ is combinatorial. Then with probability at least $1-\frac1{100}|\Xi|^{-|C|}$, a set $U$ drawn according to the $\alpha \cdot n^{-1/q^2}$-sampling distribution satisfies $|\gamma_{\sigma,U,\mathcal W}-\gamma_{\sigma,[n],\mathcal W}|\leq\frac1{10}$.
\end{lemma}

\begin{proof}
Let us first analyze which members of $\mathcal W$ get into $\mathcal W_U$. Since $\{Q\setminus C:Q\in\mathcal W\}$ is a family of disjoint sets of size $i$, the choice of $U$ means that every $Q\in\mathcal W$ becomes a member of $\mathcal W_U$ with probability exactly $\alpha^i \cdot n^{-i/q^2}$, independently of other members of $\mathcal W$. We now refer to Lemma  \ref{lem:dev} (where $(\gamma_1,\ldots,\gamma_m)$ there are the satisfaction values $S((w_{\sigma,C})_Q)$ for $(Q,S)\in\mathcal F$ such that $Q\in\mathcal W$), which implies that the probability of having $|\gamma_{\sigma,U,\mathcal W}-\gamma_{\sigma,[n],\mathcal W}|>\frac1{10}$ is less than $e^{-10^{-3}\alpha^i \cdot n^{-i/q^2} \cdot \twoSidedPomCard}$. Calculation \ref{calc:devuse} bounds it by $\frac1{100}|\Xi|^{-q^{10}\log(|\Xi|)(\log(\log(|\Xi|)/\epsilon))^2n^{1-i/q}}\leq \frac1{100}|\Xi|^{-|C|}$.
\end{proof}

From the above lemma we formulate a way of approximating all pompoms in a \good\ set $J$, assuming that we have knowledge of $J$, the common core set $C$, and of course the original combinatorial test $(\mathcal F,\mu)$.

\begin{lemma}\label{lem:approxall}
Assume that $q\geq 3$, $n>\TwoSidednLB$ and $w\in\Xi^n$. Let $P=(\mathcal F,\mu)$ be a combinatorial $(\epsilon/2,1/4q^3,q,q^9\log(|\Xi|)(\log(\log(|\Xi|)/\epsilon))^2n)$-test for $(L',L)$ which is $\epsilon$-nontrivial, let $J$ be a \good\ set of $i$-pompoms for it with core $C$, and let $U$ be chosen by the $\alpha \cdot n^{-1/q^2}$-sampling distribution where $\alpha = \alphaValue$.
With probability at least $\frac1{10}$, for every $\sigma\in\Xi^{|C|}$ it holds that $|\frac1{|J|}\sum_{\mathcal W\in J}\gamma_{\sigma,U,\mathcal W}-\frac1{|J|}\sum_{\mathcal W\in J}\gamma_{\sigma,[n],\mathcal W}|\leq\frac15$.
\end{lemma}

\begin{proof}
For a fixed $\sigma\in\Xi^{|C|}$, by using Lemma \ref{lem:approxone} and Markov's inequality, we obtain that with probability at most $\frac1{10}|\Xi|^{-|C|}$ we have more than $\frac1{10}|J|$ instances $\mathcal W\in J$ for which $|\gamma_{\sigma,U,\mathcal W}-\gamma_{\sigma,[n],\mathcal W}|>\frac1{10}$. Therefore, with probability at least $\frac1{10}|\Xi|^{-|C|}$ (noting that every $\gamma$ value is always between $0$ and $1$) the following holds:
$$|\frac1{|J|}\sum_{\mathcal W\in J}\gamma_{\sigma,U,\mathcal W}-\frac1{|J|}\sum_{\mathcal W\in J}\gamma_{\sigma,[n],\mathcal W}| \leq \frac1{|J|}\sum_{\mathcal W\in J}|\gamma_{\sigma,U,\mathcal W}-\gamma_{\sigma,[n],\mathcal W}| \leq \frac1{10}+\frac1{10} = \frac15$$
A union bound over the bad events for every possible $\sigma\in\Xi^n$ concludes the proof.
\end{proof}

We now show that a constellation as defined in Definition \ref{def:constellation} implies a \good\ set of pompoms, just as for the $1$-sided case we used it to find a \revealing\ set of pompoms. Later we will use Definition \ref{def:SCM} and Lemma \ref{lem:SCMconstellation} to find the required constellation, just as we did for the case of $1$-sided tests.

\begin{lemma}\label{lem:2SidedMessyToSetOfPoms}
Let $i\in [q]$ and $n > \TwoSidednLB$, $(L',L)$ be $\epsilon$-nontrivial, $P=(\mathcal F,\mu)$ be a combinatorial $(\epsilon/2,1/4q^3,q,q^9\log(|\Xi|)(\log(\log(|\Xi|)/\epsilon)))^2n)$-test for $(L',L)$, and let $(C,\mathcal S)$ be a $(q^{10}\log(|\Xi|)(\log(\log(|\Xi|)/\epsilon))^2,i)$-constellation for $\mu$.
Then there exists a set $J$ of $i$-\Pom s that is \good\ for $P$ with core $C$.
\end{lemma}

\begin{proof}
We extract pompoms from $\mathcal S$ greedily. We claim that as long as $\mu(\mathcal S)>\frac1{2(q+1)}$, we can extract an $i$-pompom $\mathcal W$ of size $\twoSidedPomCard$ with center $C$ from $\mathcal S$, which we then subtract from $\mathcal S$ and make into a new member of $J$. Assuming the claim holds, the process stops only when $J$ becomes such that Item \ref{item:GoodI} of Definition \ref{def:good} holds, because we started with a set $\mathcal S$ of weight at least $\frac1{q+1}$. Also, Item \ref{item:GoodCoreSize} of Definition \ref{def:good} follows from Condition \ref{cond:MessyCore} of Definition \ref{def:constellation} (regarding $\mathcal S$ and $C$), while Item \ref{item:GoodPomSize} follows from the construction described above.

It thus remains to show the following claim: Given a set $\mathcal S'\subseteq \mathcal S$ for which $\mu(\mathcal S')>\frac1{2(q+1)}$ where $(\mathcal S,C)$ is a $(q^{10}\log(|\Xi|)(\log(\log(|\Xi|)/\epsilon))^2,i)$-constellation for $\mu$, an $i$-pompom $\mathcal W\subseteq\mathcal S'$of size $\twoSidedPomCard$ with center $C$ exists.

Since $\mu(\mathcal S') \geq \frac{1}{2(q+1)}$, by Lemma~\ref{lem:lowprob}, we have that $|\bigcup_{Q\in \mathcal S'}Q| \geq \epsilon n/2$.
Defining $\mathcal B'=\{Q\setminus C:Q\in\mathcal S'\}$, we observe that
$|\bigcup_{Q\in\mathcal B'}Q| \geq |\bigcup_{Q\in\mathcal S'}Q| - |C| \geq \epsilon n/3$, because $n > \TwoSidednLB$ and so $|C| < \epsilon n/6$.
Thus, by Condition~\ref{cond:MessyInter} of Definition \ref{def:constellation}, there exist a disjoint family $\mathcal V\subseteq\mathcal B$ of $\twoSidedPomCard$ sets. Thus the family $\mathcal W=\{Q:Q\setminus C\in\mathcal V\}$ is the required $i$-pompom for the claim.
\end{proof}

Now we prove the main result of this section.

\begin{proof}[Proof of Theorem \ref{thm:2sided}]
Given a combinatorial $(\epsilon/2,1/4q^3,q,q^9\log(|\Xi|)(\log(\log(|\Xi|)/\epsilon)))^2n)$-test for $(L',L)$,
where $n>\TwoSidednLB$, we construct for $\alpha=\alphaValue$ a $2$-sided $(\epsilon,\frac1{10})$-test for $(L',L)$ that uses the $\alpha\cdot n^{-1/q^2}$-sampling distribution.

First we construct for the distribution $\mu$ of the combinatorial test the families $\Sets{i}$, $\Core{i}$ and $\Match{i}$ as per Definition \ref{def:SCM}. Similarly to the proof of Theorem \ref{thm:1sided}, every set in $\Match{0}$ is in particular a subset of $\Core{0}$, and so the union of all the sets in $\Match{0}$ is also a subset of $\Core{0}$. Therefore, by Item \ref{item:CCoreSize0} of Observation \ref{obs:SCMcombi}, $|\bigcup_{Q\in \Match{0}}Q| \leq |\Core{0}| < q^{10}\log(|\Xi|)(\log(\log(|\Xi|)/\epsilon))^2n^{1-1/q} < \epsilon n/2$, where the last inequality follows from  $n > \TwoSidednLB$. Therefore, by Lemma \ref{lem:lowprob}, $\mu(\Match{0}) \leq 1/(q+1)$, so by Lemma \ref{lem:SCMconstellation} there exists a $(q^{10}\log(|\Xi|)(\log(\log(|\Xi|)/\epsilon))^2,i)$-constellation for some $i\in Q$.

Moreover, we note that such a constellation can indeed be computed from only the knowledge of $\supp(\mu)$. We set $(\mathcal S,C)$ to be such a constellation, and then use Lemma \ref{lem:2SidedMessyToSetOfPoms} (which is also constructive) to obtain the \good\ set $J$ of $i$-pompoms.

The test proceeds as follows. Given the set $U$ produced by the $\alpha\cdot n^{-1/q^2}$-sampling distribution, we query all of it. Then, for every $\sigma\in\Xi^{|C|}$ and every $\mathcal W\in J$, we calculate $\gamma_{\sigma,U,\mathcal W}$ using our queries, and then calculate $\gamma_{\sigma,U}=\frac1{|J|}\sum_{\mathcal W\in J}\gamma_{\sigma,U,\mathcal W}$ for every $\sigma$. If there was a $\sigma\in\Xi^{|C|}$ for which $\gamma_{\sigma,U}>\frac12$, then we accept the input, and otherwise we reject it.

It remains to prove that this is indeed a correct test for $(L',L)$. Set $\mathcal I=\bigcup_{\mathcal W\in J}\mathcal W$ as per Definition $\ref{def:good}$. Since $\mu(\mathcal I)\geq \frac1{2(q+1)}$, if $u$ is any word for which the original test was $1/4q^3$-sure about, then the conditioning of the test to the set of constraints corresponding to the members of $\mathcal I$ will be $\frac1{10}$-sure for $u$ by Observation \ref{obs:condsure}. Now, since $\mu$ is uniform over its support, for any $\sigma\in\Xi^{|C|}$, the satisfaction of the original test conditioned on $\mathcal I$ by $w_{\sigma,C}$ is identical to the average $\gamma_{\sigma,[n]}=\frac1{|J|}\sum_{\mathcal W\in J}\gamma_{\sigma,[n],\mathcal W}$. In turn, Lemma \ref{lem:approxall} guarantees that with probability at least $\frac9{10}$, for all such $\sigma$ we have $|\gamma_{\sigma,U}-\gamma_{\sigma,[n]}|\leq\frac15$. Assume from now on that this event has indeed occurred.

If $w$ was a word in $L'$, then the original test accepted it with probability at least $1-1/4q^3$, and hence for $\sigma=w_C$ (for which $w_{\sigma,C}=w$) we have $\gamma_{\sigma,[n]}\geq\frac9{10}$ and hence $\gamma_{\sigma,U}>\frac12$, and the sampling test will accept on account of this $\sigma$.

On the other hand, if $w$ was a word $\epsilon$-far from $L$, then for every $\sigma\in\Xi^{|C|}$, the word $w_{\sigma,C}$ is $\epsilon/2$-far from $L$ (recall that in particular $|C|<\epsilon n/2$), and so the original test will accept it with probability at most $1/4q^3$. Hence $\gamma_{\sigma,[n]}\leq\frac1{10}$ for every such $\sigma$, and hence $\gamma_{\sigma,U}<\frac12$. This means that the sampling test will reject, as there will be no $\sigma$ on whose account the test can accept.
\end{proof}

\section{Implications of our results}
\newcommand{\unionrone}{2^{(\alpha_1)^{-1}n^{(q^{-2}-\gamma)}-1}}
\newcommand{\unionrtwo}{2^{(10\alpha_2)^{-1}n^{(q^{-2}-\gamma)}}}

The following corollaries result respectively from Theorem \ref{thm:1sided} and Theorem \ref{thm:2sided}, considering that a multitest scheme (as described in the introduction) immediately leads to a test for a union of the properties.

\begin{corollary}\label{cor:Union1Sided}
Let $q=30q'\log(q')/(1-\delta)$ and $\alpha_1 =\log(|\Xi|)q^3/\epsilon$.
For every $n>(24q(q+1)^2(\log(|\Xi|))^2/\epsilon)^{q}$,
 if $L\subseteq \Xi^n$ is the union of $r$ properties
 $L_1,\cdots,L_r$ each having $1$-sided $(\epsilon/2,\delta,q')$-tests where
 $r\le \unionrone$, then $L$ has a non-adaptive $1$-sided
 $(\epsilon,1/2)$-test with query complexity $O(n^{1-\gamma})$.
\end{corollary}

\begin{proof}
  First, we use Lemma~\ref{lem:effective1sided} to convert the
  $(\epsilon/2,\delta,q')$-test for every $L_i$ to a non-adaptive
  $1$-sided $(\epsilon/2,1/2(q+1),q,4(q+1)^2\log(|\Xi|)n)$-test where
  $q=30q'\log(q')/(1-\delta)$. We then use Theorem~\ref{thm:1sided} to
  obtain an $\alpha_1 n^{-1/q^2}$-sampling $1$-sided $(\epsilon/2,1/2)$-test for
  every $L_i$. We then amplify the probability, by repeating the test and
  rejecting if any of the runs rejects, to
  obtain a $\log(2r)\alpha_1 n^{-1/q^2}$-sampling $1$-sided
  $(\epsilon/2,1/2r)$-test for every $L_i$. We construct a multitest
  for $L_1,\ldots,L_r$, which reuses the same queries for each sample-based test,
  and from it derive the test for $\bigcup_{i=1}^rL_i$. Since $r$ is at most
  $\unionrone$, this gives a non-adaptive $1$-sided
  $(\epsilon,1/2)$-test with query complexity $O(n^{1-\gamma})$.
\end{proof}
 
\begin{corollary}\label{cor:Union2Sided}
  Let $q=20q'\log(q')\log\log(\log(|\Xi|/\epsilon))/(\frac{1}{2}-\delta)^2$
  and let $\alpha_2 =\alphaValue$.
For every $n>\TwoSidednLB$,
   if a property $L\subseteq\Xi^n$ is the union of $r$ properties
   $L_1,\cdots,L_r$ each having a $2$-sided $(\epsilon/2,\delta,q')$-test
   where $r\le \unionrtwo$, then $L$ has a non-adaptive
   $2$-sided $(\epsilon,1/10)$-test with query complexity $O(n^{1-\gamma})$.
\end{corollary}

\begin{proof}
  First, use Lemma~\ref{lem:effective2sided} to convert the
  $(\epsilon/2,\delta,q')$-test for each $L_i$ to a combinatorial $2$-sided  
  $(\epsilon/2,1/(4q)^3,q,q^9\log(|\Xi|)(\log(\log(|\Xi|)/\epsilon))^2
n)$-test, where
$q=20q'\log(q')\log\log(\log(|\Xi|)/\epsilon)/(\frac12-\delta)^2)$. Then we
use Theorem~\ref{thm:2sided} to convert each of these tests to an $\alpha_2
n^{-1/q^2}$-sampling $2$-sided $(\epsilon,1/10)$ test. Now, we convert them
to $10\log (r)\alpha_2 n^{-1/q^2}$-sampling $2$-sided $(\epsilon,1/10r)$-tests by
repeating each test $10\log r$ times independently and taking the majority
vote. We construct a multitest
for $L_1,\ldots,L_r$, which reuses the same queries for each sample-based test,
and from it derive the test for $\bigcup_{i=1}^rL_i$. Since $r$ is at most $\unionrtwo$, this gives the
$2$-sided $(\epsilon,1/10)$-test with query complexity $O(n^{1-\gamma})$.
\end{proof}


\begin{definition}[following Definition 2.1 of~\cite{MAP}]
A {\em Merlin-Arthur proof of proximity} ($\mathcal{MAP}$) for a property
$L\subseteq \Xi^n$, with proximity parameter $\epsilon$, query complexity
$q$ and proof complexity $p$, consists of a probabilistic algorithm $V$,
called the verifier, that is given a proof string $\pi\in \Xi^p$; in
addition, it is given oracle access to a word $w\in\Xi^n$, to which it is
allowed to make up to $q$ queries. The verifier satisfies the following two
conditions:
\begin{enumerate}
\item {\em Completeness}: For every $w \in L$, there exists a string
  $\pi\in \Xi^p$ (referred to as a {\em proof} or {\em witness}) such that
  $\Pr[V(n,\epsilon,w) = 1] \geq 2/3$.
\item {\em Soundness}: For every $w \in\Xi^n$ which is $\epsilon$-far from
  $L$, and any $\pi\in\Xi^p$, $\Pr[V (n,\epsilon,w) = 1]\leq 1/3$.
\end{enumerate}
If the completeness condition holds with probability $1$, then we say that
the $\mathcal{MAP}$ has $1$-sided error, and otherwise we say that it has
$2$-sided error. Also, we may say that it is {\em non-adaptive} if it makes its
queries to $w$ based only on $\pi$, before receiving any responses from
$w$.
\end{definition}

For our purposes,  we note that the proof of a $\mathcal{MAP}$ scheme for a property $L$ induces a decomposition of $L$ into sets whose union is $L$, each admitting a corresponding partial testing algorithm.
Specifically, for every $w\in L$ we define $\Pi_w$ to be any non-empty
subset of the set of proofs $\pi\in\Xi^p$ that make the verifier accept $w$
with the required probability. Then, for every $\pi\in\Xi^p$ we set
$L_\pi=\{w\in L:\pi\in\Pi_w\}$ (it may be the case that some $L_\pi$ are
empty).

Under this interpretation, for a word in the property, the proof $\pi$ is simply an indicator that the word belongs to $L_{\pi}$.
Thus, the verifier of the $\mathcal{MAP}$ scheme can be seen as receiving as input a proof $\pi$ and then running a partial test for $(L_\pi,L)$.
Consequently, the existence of a $\mathcal{MAP}$ scheme with query complexity $q$ and proof complexity $p$ for a property $L$ is the same as having a family of $|\Xi|^p$ properties $\{L_\pi:\pi\in\Xi^p\}$ such that $L=\bigcup_{\pi\in\Xi^p}L_\pi$, and there exists a partial test for  every pair $(L_\pi,L)$.


Similarly to Corollary \ref{cor:Union2Sided}, only using the validity of Theorem \ref{thm:2sided} for partial tests as well, we obtain:

\begin{corollary}\label{cor:Map2sided}
Let $q=20q'\log(q')\log\log(\log(|\Xi|/\epsilon))/(\frac{1}{2}-\delta)^2$
and let $\alpha_2 =\alphaValue$.
For every $n>\TwoSidednLB$,
 if a property $L\subseteq \Xi^n$ has a non-adaptive $2$-sided
 $(\epsilon/2,1/10)$-test with query complexity of $\Omega(n^{1-\gamma})$,
 then every $2$-sided $\mathcal{MAP}$ scheme for $L$, that has query
 complexity $q'$,
 has proof complexity $\Omega(n^{q^{-2} - \gamma}/10\alpha_2)$.
\end{corollary}

Although Theorem \ref{thm:1sided} was stated and proved for (non-partial) $1$-sided tests only, it can also be made to work for partial tests, and to give a corollary with an improved bound for this case.

\begin{corollary}\label{cor:Map1sided}
Let $q=30q'\log(q')/(1-\delta)$ and let $\alpha_1 =\log(|\Xi|)q^3/\epsilon$.
For every $n>(24q(q+1)^2(\log(|\Xi|))^2/\epsilon)^q$,
 if a property $L\subseteq \Xi^n$ has a non-adaptive $1$-sided
 $(\epsilon/2,1/2)$ query complexity of $\Omega(n^{1-\gamma})$,
 then every $1$-sided MAP scheme for $L$, that has query complexity $q'$,
 has proof complexity $\Omega(n^{q^{-2} - \gamma}/\alpha_1-1)$.
\end{corollary}

We note some concrete applications of the above results.

\begin{itemize}
\item In~\cite{SpaceVSQC}, it was shown that there exists a language $LS$
  with logarithmic space complexity that  satisfies the following: every
  non-adaptive $2$-sided  $(\epsilon/2,\delta)$-test for $LS\cap \{0,1\}^n$
  has query complexity $\Omega(n)$.  By Corollary~\ref{cor:Union2Sided},
  this means that for every large enough $n$, $L$ is not the union of less
  than $\unionrtwo$ properties over
  $L^*\subseteq \Xi^n$ each having a $2$-sided
  $(\epsilon/2,\delta,q')$-test where
  $q=20q'\log(q')\log\log(\log(|\Xi|/\epsilon))/(\frac{1}{2}-\delta)^2$ .
  By Corollary~\ref{cor:Map2sided}, $LS$ does not have a $2$-sided
  $\mathcal{MAP}$ with query complexity $q'$ and proof complexity
  $o(n^{q^{-2} - \gamma}/10\alpha_2)$.
Similarly, such conclusions apply to the properties of the small CNF formula that were studied in~\cite{3CNF}.
\item Our result also applies to the sparse graph property of $3$-colorability, which in~\cite{3Coloring} is shown to  have a linear $2$-sided test query complexity. Note that in the sparse graph model the size of the alphabet $\Xi$ is $n$, but this is still small enough for our results to provide non-trivial conclusions against decomposability.
\item 
According to the results  in \cite{SmallWidthBP} and \cite{FormulaSat} respectively, every property defined by a constant width read-once branching program or a constant arity read-once Boolean formula is testable.
Hence our results imply that properties whose testing requires $\Omega(n^{1-\gamma})$ many queries cannot be written as the union of a few properties that have such representations.
\end{itemize}

\bibliographystyle{plain}
\bibliography{Universal}

\appendix

\section{Calculations}\label{app:calc}

This appendix is reserved for calculations that are too long and bothersome to be put where they are originally used.

\begin{calculation}\label{calc:withi}
For an alphabet $\Xi$, positive integers $q$ and $n>(24q(q+1)^2(\log(|\Xi|))^2/\epsilon)^q$, any $0<\epsilon\leq 1$, $i\in [q]$, and $\alpha=15\ln{|\Xi|}\cdot q(q+1)^2/\epsilon$, we write:

$$ \left(1- \alpha^i\cdot n^{-i/q^2}\right)^{(1/3i)\epsilon\cdot n^{1-(i-1)/q}}  \leq 
\left(e^{-\alpha^i\cdot n^{-i/q^2}}\right)^{(1/3i)\epsilon\cdot n^{1-(i-1)/q}}
 = e^{-\alpha^i\cdot n^{-i/q^2}\cdot(1/3i)\epsilon\cdot n^{1-(i-1)/q}}$$

If $i\geq 2$ then we can clearly bound this so:

$$e^{-\alpha^i\cdot n^{-i/q^2}\cdot(1/3i)\epsilon\cdot n^{1-(i-1)/q}}
 <  e^{-\ln{|\Xi|}\cdot5q(q+1)^2\log(|\Xi|) n^{1-(i-1)/q-i/q^2}}  < \frac12\cdot |\Xi|^{-4q(q+1)^2\log(|\Xi|) n^{1-i/q}}$$

For $i=1$, we use the lower bound on $n$ to show $n^{1-1/q^2}>\log(|\Xi|)n^{1-1/q}$, and so:

$$e^{-\alpha\cdot n^{-1/q^2}\cdot(1/3)\epsilon\cdot n^{1-(1-1)/q}}
 <  e^{-\ln{|\Xi|}\cdot5q(q+1)^2\log(|\Xi|) n^{1-1/q}}  < \frac12\cdot |\Xi|^{-4q(q+1)^2\log(|\Xi|) n^{1-1/q}}$$
\end{calculation}

\begin{calculation}\label{calc:devuse}
For $n>(24q^{10}(\log(|\Xi|))^2/\epsilon)^q$, $\alpha = 10^3\ln(|\Xi|)\cdot q^4/\epsilon$, $q\geq 3$ and $i\in [q]$, for the case $i\geq 3$ we write:
$$e^{-10^3\alpha^i\cdot\epsilon\cdot n^{1-(i-1)/q-i/q^2}\cdot\frac1{3i}} \leq e^{-\ln(|\Xi|)^3q^{12}\cdot \epsilon^{-2}\cdot n^{1-(i-1)/q-i/q^2}\cdot\frac1{3q}} \leq \frac{1}{100}|\Xi|^{-q^{10}\log(|\Xi|) (\log(\log(|\Xi|/\epsilon)))^2 n^{1-i/q}}$$

For $i=1$, we use $n^{1/q-i/q^2} > ( 24q^{10} (\log(|\Xi|))^2/\epsilon )^{1 - 1/q} \geq 8q^6 (\log(|\Xi|))^{4/3}/\epsilon^{2/3}$, and for $i=2$ we use $n^{1/q-i/q^2} > ( 24q^{10} (\log(|\Xi|))^2/\epsilon )^{1 - 2/q} \geq 2q^3 (\log(|\Xi|))^{1/3}$. In both cases we substitute the value of $\alpha^i$ and write:
$$e^{-10^3\alpha^i\cdot \epsilon\cdot n^{1-(i-1)/q-i/q^2}\cdot\frac1{3i}} = e^{-10^3\alpha^i\cdot \epsilon\cdot n^{1/q-i/q^2}\cdot n^{1-i/q}\cdot\frac1{3i}} \leq \frac{1}{100}|\Xi|^{-q^{10}\log(|\Xi|) (\log(\log(|\Xi|/\varepsilon)))^2 n^{1-i/q}}$$
\end{calculation}

\end{document}